\documentclass[11pt]{article}

\usepackage{amssymb} 
\usepackage{amsmath}
\usepackage{mathtools}  

\usepackage{endnotes}

\usepackage{graphicx} 

\usepackage{tikz}
\usetikzlibrary{arrows.meta} 
\usetikzlibrary{shapes.misc}

\newlength\figheight
\newlength\figwidth

\usepackage[footnote,draft,silent,nomargin]{fixme} 
\usepackage{placeins}  
\usepackage{float}
\usepackage{xcolor}

\usepackage{natbib}
\usepackage{url}

\usepackage{hyperref}

\usepackage[toc,page]{appendix} 

\usepackage{relsize}

\usepackage{setspace}
\usepackage[font=footnotesize,labelfont=bf]{caption}  

\usepackage{amsthm} 

\theoremstyle{plain} 
\newtheorem{Theorem}{Theorem}[part]
\newtheorem{Definition}[Theorem]{Definition}
\newtheorem{Proposition}[Theorem]{Proposition}

\newtheorem{Lemma}[Theorem]{Lemma}
\newtheorem{Corollary}[Theorem]{Corollary}

\theoremstyle{plain} 
\newtheorem{Remark}[Theorem]{Remark}

\theoremstyle{definition} 
\newtheorem{Example}[Theorem]{Example}
\newtheorem{Specification}[Theorem]{Specification}

\newcommand{\cadlag}{c\`{a}dl\`{a}g}

\newcommand{\Levy}{L\'{e}vy}

\newcommand{\e}{\text{\normalfont e}} 
\def \R{\mathbb{R}}  

\def \Fc{\mathcal{F}} 

\def \M{\mathbb{M}}
\def \P{\mathbb{P}}
\def \Q{\mathbb{Q}}

\def \E{\mathbb{E}}
\def\EP{\mathbb{P}}\def\EP{\mathbb{E}^\mathbb{P}}
\def\EQ{\mathbb{Q}}\def\EQ{\mathbb{E}^\mathbb{Q}}

\def\EM{\mathbb{M}}\def\EM{\mathbb{E}^\mathbb{M}}

\def\Cov{{\rm Cov}}

\def\Var{\text{Var}}

\renewcommand{\Re}{\operatorname{Re}}

\DeclareRobustCommand{\rchi}{{\mathpalette\irchi\relax}}
\newcommand{\irchi}[2]{\raisebox{\depth}{$#1\chi$}} 
\DeclareRobustCommand{\Chi} {\scalebox{1.2}{$\rchi$} }

\newcommand{\infT}{\overline{T}}
\newcommand{\infTM}{\hat{T}}

\DeclarePairedDelimiter\norm{\lVert}{\rVert}
 \DeclarePairedDelimiter\abs{\lvert}{\rvert}
 \DeclarePairedDelimiterX\innerp[2]{\langle}{\rangle}{#1,#2} 

\DeclareMathOperator{\sign}{sign} 
\newcommand{\stochExp}{\mathcal{E}} 

\newcommand\given[1][]{\:#1\lvert\allowbreak\:\mathopen{}} 

\newcommand{\dontbreak}[1]{{}$\kern-2\mathsurround${}  \binoppenalty10000 \relpenalty10000 #1{}$\kern-2\mathsurround${}} 

\newcommand*\diff{\mathop{}\!\mathrm{d}}

\newcommand*\diffIN{\mathop{}\!\mathrm{d}}

\newcommand{\intspace}{\mathop{}\!}
\newcommand{\dd}{\mathrm{d}}
\newcommand{\quoteStrike}{\tilde{k}}

\newcommand{\normcdf}{\Phi}

\newcommand {\1}[1] {  {\bf 1 } \{#1\}}

\makeatletter
\newcommand*{\centerfloat}{%
  \parindent \z@
  \leftskip \z@ \@plus 1fil \@minus \textwidth
  \rightskip\leftskip
  \parfillskip \z@skip}
\makeatother

\newcommand{\suptext}[1] {   { \mathrm{{#1}}}}
\newcommand{\subtext}[1] {   {\mathrm{#1} } }

\newcommand{\fakeRate}{r}

\usepackage{xcolor,comment}

\excludecomment{notes}

\usepackage{mathrsfs}
\DeclareMathAlphabet{\mathpzc}{OT1}{pzc}{m}{it}
\newcommand{\assumExp}{\mathpzc{E}}

\title{Rational Models for Inflation-Linked Derivatives }
\author{Henrik  Dam $ ^{1},$ Andrea Macrina $ ^{2,3}$, David Skovmand$^1$, David Sloth $ ^{4}$
\\\\
$ ^1 $ Department of Mathematical Sciences, University of Copenhagen, Denmark. \vspace{4pt}\\
$ ^2 $ Department of Mathematics, University College London, United Kingdom.\vspace{4pt}\\
$ ^3 $ African Institute of Financial Markets \& Risk Management \\ University of Cape Town, South Africa.\vspace{4pt}\\
$ ^4$ Danske Bank, Copenhagen, Denmark.
}
\begin{document}  
\thispagestyle{empty}
\date{\today}
\maketitle
\vspace{-1cm}
\begin{abstract}
We construct models for the pricing and risk management of inflation-linked derivatives. The models are rational in the sense that linear payoffs written on the consumer price index have prices that are rational functions of the state variables. The nominal pricing kernel is constructed in a multiplicative manner that allows for closed-form pricing of  vanilla inflation products suchlike zero-coupon swaps, year-on-year swaps, caps and floors,  and the exotic limited-price-index swap. We study the conditions necessary for the multiplicative nominal pricing kernel to give rise to short rate models for the nominal interest rate process. The proposed class of pricing kernel models retains the attractive features of a nominal multi-curve interest rate model, such as closed-form pricing of nominal swaptions, and it isolates the so-called inflation convexity-adjustment term arising from the covariance between the underlying stochastic drivers. We conclude with examples of how the model can be calibrated to EUR data.\footnote{The authors are grateful to L.~P.~Hughston and to participants of the JAFEE 2016 Conference (Tokyo, August 2016), the CFE 2016 Congress (Sevilla, December 2016), 2016 QMF Conference (Sydney, December 2016), and of the London Mathematical Finance Seminar held at King's College London (30 November 2017) for comments and suggestions. The authors especially acknowledge high-quality feedback provided by two anonymous reviewers.}
\\ \vspace{-0.2cm}
\\ \vspace{-0.2cm}
\end{abstract}
\section{Introduction} \label{RISecIntro}
The inflation market has grown in the aftermath of the 2008 financial crisis. Central banks have been conducting aggressive quantitative easing to keep inflation off the cliff of deflation, and the ensuing fears have driven hedging needs. As a consequence, the market for trading inflation has soared to the point where standard inflation derivatives are now cleared on the London Clearing House (LCH) in numbers exceeding 100 bn EUR measured by notional outstanding value in early 2017. As this number only counts linear derivatives, the total market size is likely much larger. Among the products cleared one finds the Year-on-Year swap (YoY swap), swapping annual inflation against a fixed strike, and the Zero-Coupon swap (ZC swap), which swaps cumulative inflation against a fixed strike at maturity.

 Among the OTC-traded nonlinear derivatives, the most important is arguably the YoY cap/floor, which is in principle a portfolio of calls (caplets) or puts (floorlets) with equal strike on YoY inflation. Another significant derivative is the ZC cap/floor, which is simply a call/put on the ZC swap rate. The derivatives market is dwarfed in size by the market for inflation-linked bonds. These bonds are typically government-issued debt where the principal is linked to the consumer price index (CPI) or similar. The bonds often have an embedded YoY floor protecting the principal from being adjusted downwards by deflation. Limited Price Index  (LPI) products come with both a lower and upper bound on the principal adjustment creating a path-dependent collar on inflation. Despite its exotic nature LPIs have been in high demand by pension funds. All products should ideally be priced in a consistent manner using a tractable arbitrage-free model. Cap/floor products display volatility skews and non-flat term structures of volatility, both of which the model also should be able to capture. Besides, the model should yield closed-form solutions for the price of the most traded derivatives, here the YoY and the ZC cap/floor.

\citep{Hughston1998} develops a general arbitrage-free theory of interest rates and inflation in the case where the consumer price index and the real and nominal interest rate systems are jointly driven by a multi-dimensional Brownian motion. This approach is based on a foreign exchange analogy in which the CPI is treated as a foreign exchange rate, and the ``real'' interest rate system is treated as if it were the foreign interest rate system associated with the foreign currency. The often-cited work by \citep{JarrowYildirim2003} makes use of such a setup. They consider a three-factor model (i.e., driven by three Brownian motions) in which the CPI is modelled as a geometric Brownian motion, with deterministic time-dependent volatility and the two interest rate systems are treated as extended Vasicek-type (or Hull-White) models. Similar to \citep{JarrowYildirim2003},  \citep{KainthDodgson2006}  use a short-rate approach where the nominal and the inflation rates are both modelled by Hull-White processes while discarding the idea of a real economy. A GBM-based model for the CPI provides the baseline framework for how one might understand implied volatility in such a market, but any GBM model for the CPI does not, by construction, reproduce volatility smiles.

Further development of inflation models has paralleled that of interest rates models. For example inflation counterparts to the nominal LIBOR Market Model, see for example \citep{BrigoMercurio2007}, have been studied in \citep{BelgradeEtAl2004}, \citep{Mercurio2005}, and \citep{MercurioMoreni2006}. While these models can reproduce smiles---augmented with stochastic volatility or jumps---they rely on numerically intensive algorithms or approximations for the pricing of ZC cap/floors, in particular. One may say similarly of the models by \citep{Kenyon2008}, \citep{GretarssonEtAl2012}, and \citep{MercurioMoreni2009} who in a similar manner use forward inflation, or in the case of \citep{Hinnerich2008} the forward inflation swap rate, as the model primitive. \citep{Waldenberger2017} builds an inflation counterpart to the nominal model of \citep{GrbacEtAl2015} and \citep{Keller-ResselEtAl2013}.  One also finds \citep{Ribeiro2013} in the local volatility context, \citep{Kruse2011} extending the GBM methodology with \citep{Heston1993} stochastic volatility, and \citep{SingorEtAl2013}   adding stochastic volatility to the \citep{JarrowYildirim2003} model.
Our work is inspired by the approach to nominal term structure of interest rates based on the so-called {\it rational models}. This choice is motivated by the success of the rational model framework as documented in the comprehensive empirical study of \citep{FilipovicLarssonTrolle2017} who demonstrate that linear-rational models perform as well or better than similar affine term structure models. Furthermore, the rational model framework has been extended to model multiple nominal curves and credit risk in \citep{CrepeyEtAl2016} and \citep{MacrinaMahomed2018}; it is this approach we follow. This framework allows for analytical expressions for swaptions, which is not the case for affine term structure models. In this paper, we demonstrate how rational models for inflation are constructed, which retain the tractability of the nominal counterpart and can price, in closed-form, all the relevant derivatives suchlike YoY and ZC cap/floors and LPI swaps.
  
   In Section 2, we first present the model in full generality. Following \cite{DoberleinSchweizer2001}, we study the conditions for a short rate model representation to be obtained. In Section 3 we derive option pricing formulae under different assumptions in the driving process, and in Section 4 we end with an example that shows how the model can be simultaneously calibrated to inflation derivatives and a multiple-curve nominal interest rate market.
\section{Rational term structures} \label{RISecRatTermStruc}
We adopt the pricing kernel approach, which was pioneered  by 
\citep{Constantinides1992}, \citep{FlesakerHughston1996a}, \citep{FlesakerHughston1996b}  and \citep{Rogers1997}---for a good summary see \citep{HuntKennedy2004} and, for a more recent account, \citep{GrbacRunggaldier2016}.  \citep{MacrinaMahomed2018} propose pricing kernel models to construct so-called curve-conversion factor processes, which link distinct yield-curves in a consistent arbitrage-free manner, and which give rise to the across-curve pricing formula for consistent valuation and hedging of financial instruments across curves. Applications include the pricing of inflation-linked and hybrid fixed-income securities. A property of the pricing kernel approach is the ease with which the pricing and hedging of multiple currencies can be handled. This is the property one benefits from when considering inflation-linked pricing, and nominal and real economies are introduced in analogy to domestic and foreign economies. Compared to the classical approach, in order to allow for negative short rates, we relax the paradigm  and consider  general semimartingale dynamics for the pricing kernels. The approach taken next is one where the existence of a pricing kernel model is postulated and its dynamics are modelled. It is via the pricing kernel that no-arbitrage price processes of tradable assets are generated by imposing that the asset price process, when multiplied by the pricing kernel process, be a martingale with respect to the probability measure the pricing kernel dynamics are produced. This no-arbitrage notion is one presented in textbooks suchlike, e.g., \cite{HuntKennedy2004} and \cite{Bjork2009}.
\subsection{General model} \label{RISectGenMod}
We model a financial market by a filtered probability space $(\Omega,\Fc,\P,(\Fc_t)_{0\leq t})$, where $\P$ denotes the real probability measure and $(\Fc_t)_{0\leq t}$ the market filtration satisfying the usual conditions. A finite time horizon is considered, i.e., a time line $0\le t\le T< \infT<\infty$, throughout. 
\begin{Definition}[Pricing kernel]
We call a stochastic process $(\pi_t)_{0\leq t}$ with $\pi_0=1$ a pricing kernel if it is a strictly positive, {\cadlag}, semimartingale such that $\pi_t$ has  finite expectation for all $t\geq 0$.  
\end{Definition}
Let  $\mathcal{L}_T^1(\mu;\pi) = \{ \Chi: \Omega \to \R \text{ s.t. } \Chi \text{ is } \Fc_T\text{-measurable and }  \E^{\mu} [ \abs{ \pi_T \Chi } < \infty]\}$ where $\mu$ is a probability measure on $(\Omega,\Fc)$. Let $(\pi^{\suptext{N}}_t)_{0 \leq t }$ be the (nominal) pricing kernel process. If we consider some claim  $\Chi \in \mathcal{L}_T^1(\P;\pi^{\suptext{N}}) $, then by  standard no-arbitrage theory, see e.g. \citep{HuntKennedy2004}, the process $(V^{\chi} _t  )_{0 \leq t \leq T} $, defined by
\begin{equation} \label{RIeqPricingFormula}
V^{\chi}_{t}  = \frac{1}{\pi_t^{\suptext{N}}} \EP_t  \left[  \pi_T^{\suptext{N}} \Chi\right],
\end{equation}
is an arbitrage-free price process.  The notation $\E_t[\cdot]$ is short-hand for $\E[\,\cdot\,\vert\mathcal{F}_t]$. Following \citep{NguyenSeifried2015}[Proposition 2.2], we have:
\begin{Proposition}
Consider $n$ assets with price processes $(S^1_t),\dots,(S^n_t)$ satisfying Eq.~\eqref{RIeqPricingFormula}, i.e., such that $(\pi_t^{\suptext{N}} S_t^i )$ is a $\P$-martingale for $i=1,\dots,n$. Assume the asset with strictly positive price process $(S_t^1)$ is traded. Then, the market is free of arbitrage.
\end{Proposition}
\begin{proof}
By Eq. \eqref{RIeqPricingFormula}, the process $\xi_t := S^1_t \pi^{\suptext{N}}_t /S^1_0$ is a strictly positive martingale with $\xi_0=1$. A measure $\Q$ may be defined by $\xi_t = \left. \frac{\dd \Q}{\dd \P}  \right \lvert_{\Fc _t}$ on any finite interval, and by the Bayes' rule one obtains
\begin{equation*}
\EQ_t \left [ \frac{S^i_T}{S^1_T} \right  ] = \frac{\EP_t \left [\xi_T (S^i_T/S^1_T) \right ]}{\xi_t} = \frac{\EP_t \left [\pi^{\suptext{N}}_T S^i_T \right ]}{\pi_t S^1_t}  = \frac{S^i_t}{S^1_t},
\end{equation*}
for $2\leq i \leq n$. Thus, $\Q$ is a risk-neutral measure associated with the numeraire $(S^1_t)$. The existence of a pricing kernel, here the process $(\pi^N_t)_{0\le t}$, guarantees absence of arbitrage, also in the case of uncountably many assets. Here we refer to the no-arbitrage notion of ``no asymptotic free lunch with vanishing risk'' (NAFLVR) developed in \citep{CKT2016}.
\end{proof}
We are agnostic as to how the asset with price process $(S^1_t)$ is chosen; for example it may be a zero-coupon bond. From formula \eqref{RIeqPricingFormula} it follows that, for $0\le t\le T$, the nominal zero-coupon bond price system,
\begin{equation}  \label{RIeqNomBondPriceGen}
P^{\suptext{N}}_{tT}  = \frac{1}{\pi_t^{\suptext{N}}} \EP_t  \left[  \pi_T^{\suptext{N}} \right],   
\end{equation} 
is free of arbitrage opportunities. Assuming that $P^{\suptext{N}}_{tT}$ is differentiable in $T$, the short rate process $(r^{\suptext{N}}_t)_{0 \leq t \leq \infT}$ may be obtained by the well-known relation 
$
r^{\suptext{N}}_t = -\partial_T \ln\left(P_{tT}^{\suptext{N}}\right)\rvert _{T=t}.
$
This tells that $(\pi^{\suptext{N}}_t) $  determines simultaneously the inter-temporal risk-adjustment and the discounting rate. 

The goal is to produce models, which facilitate the pricing of inflation-linked derivatives. To this end, we equip the framework with a real-market analogous to the foreign economy in the foreign-exchange analogy. If we assume that $(\pi_t^{\suptext{R}})_{0\le t}$ is a pricing kernel for the real market, then the foreign-exchange analogy establishes the relationship
\begin{equation}\label{PI-def}
C_t = \pi_t^{\suptext{R}}/\pi_t^{\suptext{N}}
\end{equation}
where $(C_t)_{0\leq t}$ denotes the CPI process that acts as an exchange rate from the nominal to the real economy, see, e.g. \cite[Proposition 17.11]{Bjork2009}.

As in \citep{FlesakerHughston1996a}, \citep{FlesakerHughston1996b}, \citep{Rutkowski1997} and  \citep{Rogers1997},  we introduce an extra degree of flexibility and model prices with respect to an auxiliary measure  $\M$. This extra degree of freedom allows for simplified calculations or more tractable modelling under the $\M$-measure while desirable statistical properties may still be captured under the $\P$-measure. In fact it is also possible to build in  terminal distributions or ``views'' under $\P$, in the spirit of \citep{BlackLitterman1992}, and as explicitly obtained in \citep{Macrina2014}. This is a feature expected by practitioners of inflation-linked trading, motivated by the fact that inflation is an area that often receives significant attention from monetary policymakers and is subject to so-called ``forward guidance''. With regard to how to induce the measure change for such a purpose, we refer to \citep{HoyleEtAl2011}, \citep{Macrina2014} for the multivariate generalisation, and \citep{CrepeyEtAl2016} for an application in a multi-curve term structure setup. We shall model the Radon-Nikodym process $(M_t)_{0 \le t}$ with $M_0 = 1$ as a strictly positive,  {\cadlag} martingale and fix some time $\infTM<\infty$.  Then,
$
\M (A) = \EP[M_{\infTM}\, \1{A}],
$
for $A\in \Fc_{\infTM}$, defines an equivalent  measure. By setting $h^{\suptext{N}}_t =\pi_t^{\suptext{N}} / M_t $, with no loss of generality, we can express the fundamental pricing equation~\eqref{RIeqPricingFormula} under $\M$ by the Bayes formula: 
\begin{align}
V^{\chi} _t   = \frac{1}{\pi_t^{\suptext{N}}} \EP_t  \left[  \pi_T^{\suptext{N}} \Chi\right] 
=\frac{1}{h_t^{\suptext{N}} M_t} \EP_t  \left[ M_T h_T^{\suptext{N}}  \Chi\right] = \frac{1}{h_t^{\suptext{N}} } \EM_t  \left[  h_T^{\suptext{N}}   \Chi\right],   \label{RIeqPricingFormulaM}
\end{align}
for $0 \leq t\leq T  < \infT \wedge \infTM$ and $\Chi \in \mathcal{L}^1_{T} ( \P;\pi^{\suptext{N}})$.  That is, $(h_t^{\suptext{N}})$ is the nominal pricing kernel under the $\M$-measure. Similarly,  the relationship $\pi_t^{\suptext{R}} = M_t h^{\suptext{R}}_t$ introduces the real pricing kernel $(h^{\suptext{R}}_t)$ under $\M$. It follows that, under $\M$, $(h_t^{\suptext{N}})$ and $ (h_t^{\suptext{R}})$   are strictly positive, semimartingales, see \citep{JacodShiryaev2003}[III Theorem 3.13],  and that $C_t h_t^{\suptext{N}} = \, h_t^{\suptext{R}}$ for all $t\geq 0$.
\\

\noindent{\bf Modelling convention}. Let $s_t:=1/C_t$ for $t\ge 0$, for modelling convenience. From the relation \eqref{PI-def}, it then follows $h^N_t=s_t h^R_t$. We model $(h^R_t)$ and $(s_t)$, where $h^R_0=1/s_0=C_0$, as strictly positive semimartingales under $\mathbb{M}$ such that $h^R_t$ and $h^N_t$ have finite expectation for $t\ge 0$. 
\begin{Definition}[Real-kernel spread model]
Let the triplet $(h_t^{\suptext{R}},s_t,M_t)_{0\leq t}$ be such that $(h_t^{\suptext{R}})_{0\leq t }$, $(s_t)_{0\leq t }$ and $(M_t)_{0 \leq t}$   are strictly-positive, {\cadlag}, $h^R_0=1/s_0=C_0$ and $M_0=1$. Furthermore assume $(h_t^{\suptext{R}})_{0\leq t }$ and $(s_t)_{0\leq t }$  are semimartingales and that  $(M_t)_{0 \leq t}$ is a martingale. Denote by $\M$ the measure induced by $(M_t)$. Assume that $h_t^{\suptext{R}}$ and $h_t^{\suptext{R}}  s_t$ have finite expectation for all $t\geq 0$ under $\M$. We call such a triplet a real-kernel spread model (RSM).
\end{Definition}
Often, the pricing of inflation-linked instruments is performed under either the nominal risk-neutral measure $\Q ^{\suptext{N}}$ or the real risk-neutral measure $\Q ^{\suptext{R}}$. In the general setting presented so far, one is not necessarily in a position to get consistent prices under these measures. In Section~\ref{RISecComOthMod}, we treat this issue in the context of some well-known models, which use from the outset a risk-neutral measure. In Section \ref{RISecFundProd} we proceed to the pricing of primary inflation-linked securities in the backdrop of a  more specific model class. In Section~\ref{RISectMChange} we discuss the change to risk-neutral measures in the same model class.  
\subsection{Comparison with other models} \label{RISecComOthMod} 
In this section, we discuss other models and in a few cases show that our specification can be regarded as a generalisation. The comparisons shall help to understand our modelling approach in that they show how our model ingredients would look in known models.

In the case of equity pricing, the benchmark model is the geometric Brownian motion specification of \citep{BlackScholes1973}. In this sense, the most natural translation of this to inflation modelling is done by \citep{KornKruse2004} specifying the inflation index under the nominal risk-neutral measure by
\begin{equation*}
\dd C_t  = C_t  ( r^{\suptext{N}}(t)  - r^{\suptext{R}}(t) )    \dd t + C_t\sigma_C  \dd W^{\suptext{C}}_t,
\end{equation*}
where $r^{\suptext{N}}$ and $r^{\suptext{R}}$ are the deterministic nominal and real interest rates and $(W^C_t)_{0\leq t}$ is a Brownian motion. Black-Scholes-type pricing formulae are derived for ZC caps with payoff function $\max[C_T/C_0-K,0]$, and \citep{Rubinstein1991} derives a pricing formula of a similar type for YoY caplets with payoff function $\max[C_{T_i}/C_{T_{i-1}}-K,0]$. We refer to \citep{Kruse2011}  for the exact formulae. The formulae for the ZC cap and YoY caplet as functions of the volatility parameter $\sigma_C$  can be inverted to implied volatilities as it is commonly done for equity options. This will be relevant in Section~\ref{RISecCaliExam} when we calibrate some specific rational pricing models. 

\citep{JarrowYildirim2003} produce an important generalisation that allows for the pricing of inflation-linked securities with stochastic interest rates. In practice, the popular model specification is to assume that the nominal and the real interest rates have Hull-White dynamics. Under the nominal risk-neutral measure, such a model specification takes the form
\begin{align*}
&\dd r^{\suptext{N}}_t =\left [ \theta_{\subtext{N}}(t) -  a_{\subtext{N}} r^{\suptext{N}} _t \right ] \dd t + \sigma_{\subtext{N}} \dd W^{{\suptext{N}}} _t,&
&\dd r^{\suptext{R}}_t =\left [ \theta_{\subtext{R}}(t) - \rho_{{\suptext{RC}} } \sigma_{\subtext{C}} \sigma_{\subtext{R}} - a_{\subtext{R}} r^{\suptext{R}}_t  \right ] \dd t + \sigma_{\subtext{R}} \dd W^{\suptext{R}}_t,& \\
&\dd C_t  =  C_t   \left ( r_t^{\suptext{N}}  - r_t^{\suptext{R}} \right  )    \dd t + C_t  \sigma_{\subtext{C}}   \dd W^{\suptext{C}}_t,&
\end{align*}
where $(W^{\suptext{N}}_t)$, $(W^{\suptext{R}}_t)$ and $(W^C_t)$ are dependent Brownian motions, and where $\theta_{\subtext{N}}(t)$ and $\theta_{\subtext{R}}(t)$ are functions chosen to fit the term-structure of interest rates, see \citep{BrigoMercurio2007}[Chapter 15] and \citep{HullWhite1990}.

\begin{Proposition} \label{RIPropGenJY}
 The \citep{KornKruse2004} and the \citep{JarrowYildirim2003} models are RSM triplets where $(h^{\suptext{R}}_t, s_t, M_t)_{0\le t}$ is given by 
\begin{align*}
 &h^{\suptext{R}}_t=\exp\left(-\int_{0}^t r^{\suptext{R}}_s  \diff s\right) \left.  \frac{ \dd \Q^{\suptext{R}} }{\dd  \Q^{\suptext{N}}}  \right \lvert_{\Fc _t},& 
&s_t =\frac{ 1 }{h_t^{\suptext{R}}} \exp\left(-\int_{0}^t r^{\suptext{N}}_s \diff s\right),&
&M_t =\left. \frac{\dd \Q^{\suptext{N}} }{\dd \P}  \right \lvert_{\Fc _t}.&
\end{align*}
The nominal pricing kernel $(h^{\suptext{N}}_t)_{0\le t}$ is determined by $h^{\suptext{N}}_t=s_t\,h^{\suptext{R}}_t$.
\end{Proposition}
\begin{proof}
The measure change to $\Q^{\suptext{N}}$ is given in \citep{JarrowYildirim2003}[Footnote 5], the measure change to $\Q^{\suptext{R}}$ is similar and, in the Black-Scholes case, the results are standard. 
\end{proof}
The choice above is not unique. One could, e.g., set $M_t = 1$ and change the other processes accordingly, which would  amount to specifying and matching the models under $\P$, instead.

\subsection{Primary inflation-linked instruments}  \label{RISecFundProd}
We now proceed to the pricing of the primary inflation-linked products, suchlike the ZC swap and the YoY swap, which serve as the fundamental hedging instruments against inflation risk and the swap rates as underlying of exotic inflation-linked derivatives. To this end, we propose a specific class of rational pricing kernels:

\begin{Definition}[Rational pricing kernel system] \label{RIExModelCons} 
Let $\M$ be a measure equivalent to $\P$ induced by a Radon-Nikodym process $(M_t)_{0\leq t}$.  Let $(A_t^{\suptext{R}})_{0\leq t }$ and $(A_t^{\suptext{S}}
)_{0\leq t}$  be positive martingales under $\M$ with $A^{\suptext{S}}_0=A^{\suptext{R}}_0=1$.  Let  $( A_t^{\suptext{R}} A_t^{\suptext{S}})$ have finite expectation under $\M$ for all $t\geq 0$. Let the real pricing kernel $(h^{\suptext{R}}_t)_{0\le t}$ be given by  $h_t^{\suptext{R}} = R(t)\left[ 1 + b^{\suptext{R}}(t) ( A_t^{\suptext{R}} -1 ) \right]$ where $R(t)\in C^1$ is a strictly positive deterministic function with  $R(0)=C_0$, and where $b^{\suptext{R}}(t)\in C^1$ is a deterministic function that satisfies $0<b^{\suptext{R}}(t)<1$. Furthermore, let $s_t = S(t) A_t^{\suptext{S}}$
where $S(t)\in C^1$ is a strictly positive deterministic function with $S(0)=1/C_0$, and set 
$
h^{\suptext{N}}_t=s_t\, h^{\suptext{R}}_t.
$
We call $( h_t^{\suptext{R}}, s_t, M_t)_{0\leq t}$ thus specified a rational pricing kernel system (RPKS).
\end{Definition}
By Ito's lemma, $(h_t^{\suptext{R}})$ and $(s_t)$ are strictly positive semimartingales. An RPKS is in particular an RSM-triplet and therefore, by Section~\ref{RISectGenMod}, it produces a nominal and a real market, both free of arbitrage opportunities. The martingale ($A^{\suptext{R}}_t$) generates the randomness in the real market, while the joint law of ($A^{\suptext{S}}_t$) and ($A^{\suptext{R}}_t$) generates the randomness in the nominal market. All  derivations throughout will be obtained under the assumption of having an RPKS.
\begin{Proposition}[Affine payoffs evaluated in an RPKS]  \label{RIPropAffinePayoff}
Assume an RPKS. The price process $(V^{\chi}_t)_{0\le t\le T}$ of a contract with payoff function $\Chi = a_1+ a_2 C_T$, for $a_1, a_2\in \R$, at the fixed date $T\ge t\ge 0$ is given by 
\begin{equation}\label{RIeqAffpay}
V^{\chi}_t  = \frac{a_2 (b_0(T) +   b_1(T) A_t^{\suptext{R}}) + a_1 (b_2(T) A_t^{\suptext{S}} +  b_3(T)  \EM_t[  A_T^{\suptext{R}} A_T^{\suptext{S}}] ) }{ b_2(t) A_t^{\suptext{S}} + b_3(t)  A_t^{\suptext{R}} A_t^{\suptext{S}}}
\end{equation}
where,  for $0\leq t \leq T $, 
$b_0(t) = R(t)  (1-b^{\suptext{R}}(t))$,    $b_1(t) =  R(t) b^{\suptext{R}}(t)$,
$b_2(t) = R(t)  (1-b^{\suptext{R}}(t))S(t)$,  $b_{3} (t) = R(t) b^{\suptext{R}}(t)S(t)$. If $a_1=0$, i.e. the payoff is linear in $C_T$,  the price process $V^{\chi}_t$ is a rational function of $A^{\suptext{R}}_t$ and $A^{\suptext{S}}_t$.
\end{Proposition}
\begin{proof}
It follows by the $\M$-pricing equation~\eqref{RIeqPricingFormulaM}.
\end{proof}
The price process $(P^{\suptext{N}}_{tT})_{0\le t\le T}$ of the nominal ZC bond follows from Eq.~\eqref{RIeqAffpay} for $a_1=1$ and $a_2=0$. We have,
\begin{equation} \label{RIeqNomBondPrice}
P_{tT}^{\suptext{N}} =   \frac{b_2(T) A_t^{\suptext{S}} + b_3(T)  \EM_t[  A_T^{\suptext{R}} A_T^{\suptext{S}}]  }{ b_2(t) A_t^{\suptext{S}} + b_3(t)  A_t^{\suptext{R}} A_t^{\suptext{S}}}
\end{equation}
with $b_2(t)$ and $b_3(t)$ given in Proposition~\ref{RIPropAffinePayoff}. 
It then follows that the initial nominal term structure $P^{\suptext{N}}_{0t}$, $0\le t\le T$, is given by
$
P^{\suptext{N}}_{0t}=R(t)S(t)( 1+  b^{\suptext{R}} (t)(\EM[A^{\suptext{R}}_tA^{\suptext{S}}_t] -1)).
$
In particular, the parameter function $S(t)$ appearing in both, the price processes of the nominal ZC bond and the contract~\eqref{RIeqAffpay}, can thus be used for calibrating to the market-observed prices $\overline{P}{}^{\suptext{N}}_{0t}$, $0\le t\le T$, according to
$
S(t)=\overline{P}{}^{\suptext{N}}_{0t}/[ R(t)(1+  b^{\suptext{R}} (t)(\EM[A^{\suptext{R}}_tA^{\suptext{S}}_t] -1))].
$
We note that should  $t\mapsto  \EM \left[A^{\suptext{R}}_t A^{\suptext{S}}_t\right]  $ not belong to $C^1$, one can calculate its value in all relevant time points and use a $C^1$-interpolation, and nevertheless produce the same price for any financial product whose payoff only depends on state variables at those times.

The most basic inflation-linked product is the ZC swap, which gives exposure to the CPI value at the swap maturity $T$ for an annualised fixed payment. Its price process $(V_t^{\suptext{ZCS}}) _{0\leq t\leq T}$ can be written in the form
\begin{equation} \label{RIeqSwapVal}
V^{\suptext{ZCS}}_{t} = \frac{1}{h^{\suptext{N}}_t}\EM_t\left [h^{\suptext{N}}_T  \left( \frac{C_T}{C_0} -K \right)\right]  = P_{tT}^{\suptext{IL}}  - K  P_{tT}^{\suptext{N}} 
\end{equation} 
where  
$
P_{tT}^{\suptext{IL}} =\EM_t[h^{\suptext{N}}_T   C_T/C_0]/h^{\suptext{N}}_t
$
is the price of an inflation-linked ZC bond at $t\le T$. ZC swaps are highly liquid for several maturities and therefore it is reasonable to consider an actual term-structure of ZC swaps and aim at constructing models able to calibrate to the relevant market data in a parsimonious manner. By Eq.~\eqref{RIeqSwapVal}, given a nominal term-structure,  a ZC swap term-structure is equivalent to an inflation-linked ZC bond term-structure,  and fitting either is equivalent.  The price of an inflation-linked ZC bond within an RPKS follows directly from Proposition~\ref{RIPropAffinePayoff}:
\begin{equation*}
P_{tT}^{\suptext{IL}} = \frac{1}{C_0}\frac{b_0(T) +   b_1(T) A_t^{\suptext{R}}   }{ b_2(t) A_t^{\suptext{S}} + b_3(t)  A_t^{\suptext{R}} A_t^{\suptext{S}}}
\end{equation*} 
with $b_0(t)$ and $b_1(t)$ given in Proposition~\ref{RIPropAffinePayoff}.  We see that by matching the degree of freedom $R(t)$ to the initial term structure $\overline{P}{}^{\suptext{IL}}_{0t}$ of inflation-linked bonds as implied from the market, i.e. $R(t) = \overline{P}{}^{\suptext{IL}}_{0t} C_0$, the model replicates the term structure of ZC swaps. For ZC swaps, a de-annualised fair rate is quoted, namely a number $\quoteStrike$ such that for $K=(1+\quoteStrike)^T$ the initial value of the swap is zero. Given $\overline{P}{}^{\suptext{N}}_{0T}$, the initial term structure $\overline{P}^{\,\suptext{IL}}_{0T} $ is implied from the ZC swap market fair rates $\overline{k}{}^{\,\suptext{ZC}}_{0T}$ via 
$
\overline{k}{}^{\,\suptext{ZC}}_{0T}  = \left (\overline{P}{}^{\suptext{N}}_{0T}\big/\overline{P}{}^{\,\suptext{IL}} _{0T}  \right ) ^{1/T}-1.
$
The price process $(P^{\suptext{R}}_{tT})_{0\le t\le T}$ of a real ZC bond is
\begin{equation} \label{RIeqRealBond}
P^{\suptext{R}}_{tT}  = \frac{1}{h_t^{\suptext{R}}} \EM_t\left[ h^{\suptext{R}}_T\right] = \frac{ b_0(T) + b_1(T) A_t^{\suptext{R}} }{ b_0(t) + b_1(t) A_t^{\suptext{R}}  }, 
\end{equation}
with $b_0(t)$ and $b_1(t)$ as in Proposition~\ref{RIPropAffinePayoff}.  In accordance with the foreign-exchange analogy, it holds that
$
P^{\suptext{R}}_{tT}\, C_t =\ P^{\suptext{IL}}_{tT}.   
$
\

Next, we consider the Year-on-Year swap (YoY swap) which exchanges yearly percentage increments of CPI against a fixed rate. The YoY swap can be decomposed into swaplets, so we consider first the price $V_{tT_i} ^{\suptext{YoYSL}}$ at time  $t<T_{i-1}$ of a such over the period $[T_{i-1}, T_i]$. By the pricing relation~\eqref{RIeqPricingFormulaM} we have  
\begin{align}
V_{tT_i} ^{\suptext{YoYSL}} ={} & \frac{1}{h_t^{\suptext{N}}} \EM_t \left [h_{T_i}^{\suptext{N}}  \left (  \frac{C_{T_i}}{C_{T_{i-1}}}  - K \right ) \right ] 
 =P_{tT_i}^{\suptext{IL}} S(T_{i-1}) A_t^{\suptext{S}}      + \frac{   b_2(T_i)  S( T_{i-1})    \, \Cov_t^{\M}\left[ A^{\suptext{R}}_{ T_{i-1} } ,A^{\suptext{S}}_{ T_{i-1}}\right]  } { b_2(t) A_t^{\suptext{S}} + b_3(t)  A_t^{\suptext{R}} A_t^{\suptext{S}}}  -   K  P^{\suptext{N}}_{t T_i}    \label{RIeqYoYVal}
\end{align}
with $b_2(t)$ and $b_3(t)$ as in Proposition~\ref{RIPropAffinePayoff}.
For YoY swaps the fair rate $\quoteStrike$ is quoted in financial markets such that $K=1+\quoteStrike$. The price of the swap is $V_{tT_N}^{\suptext{YoYS}} = \sum_{i=1}^N V^{\suptext{YoYSL}}_{tT_i}$, from which the fair rate $\quoteStrike^{\suptext{YoY}}$ can be extracted:
\begin{equation} \label{RIEqYoYFairRate}
\quoteStrike^{\suptext{YoY}}_{tT_N} =\frac{ 1 }{    \sum_{i=1}^N  P^{\suptext{N}}_{t T_i}    }   \sum_{i=1}^N    \left(      P_{tT_i}^{\suptext{IL}}  S(T_{i-1}) A_t^{\suptext{S}}   +\frac{   b_2(T_i)  S( T_{i-1})    \, \Cov_t^{\M}\left[ A^{\suptext{R}}_{ T_{i-1} } ,A^{\suptext{S}}_{ T_{i-1}}\right]  } { b_2(t) A_t^{\suptext{S}} + b_3(t)  A_t^{\suptext{R}} A_t^{\suptext{S}}}   \right) -1. 
\end{equation}
If independence between  $(A_t^{\suptext{R}})$ and $(A_t^{\suptext{S}})$ is assumed, the YoY swap rate at time $t=0$ becomes
\begin{equation*} 
k^{\suptext{YoY}}_{0T_N} =\frac{ 1 }{    \sum_{i=1}^N  \overline{P}{}^{\suptext{N}}_{0 T_i}    }   \sum_{i=1}^N    \left ( \frac{\overline{P}{}^{\suptext{IL}}_{0T_{i-1}}}{\overline{P}{}^{\suptext{N}}_{0T_{i-1}}}     \overline{P}_{0T_i}^{\suptext{IL}}   \right  ) -1. 
\end{equation*}
Thus, if the independence assumption is imposed, the swap rate is completely determined by the inflation-linked and nominal term structures and hence can be expressed in a model-independent fashion. The difference between the market-observed swap rate and the above expression is often referred to as the convexity correction for the YoY swap of length $T_N$.

\subsection{On short-rate representation} \label{RISectMChange} 
In the following section, we study the question of the existence of a classical savings accounts in an RPKS, which is related to the work of \cite{DoberleinSchweizer2001}. We shall see that the obtained class of nominal pricing kernels can rarely be represented in terms of a short rate model. We study the reasons behind the lack of such a property. If it is possible to decompose a pricing kernel into
\begin{equation} \label{RIeqhtdecomp}
h_t = \left . \frac{\dd \Q }{\dd \M}  \right \rvert_{\Fc_t} \e^{-\int_0^t r_u \dd u}
\end{equation}
where $(r_t)$ is a short rate process and $\Q$ is the corresponding numeraire measure, then the prices obtained in either model would be equal, i.e. the same prices could instead have been obtained using a short rate approach. To investigate whether a decomposition like \eqref{RIeqhtdecomp} exists, some technical material is needed. 

We recall that a special semimartingale is a process $(X_t)$ with a unique decomposition $X_t = X_0+ B_t+ N_t$, where $(B_t)$ is predictable and of finite variation, $(N_t)$ is a local martingale and $B_0=N_0=0$. The decomposition is called the canonical (additive) decomposition. We recall \cite[II Theorem 8.21]{JacodShiryaev2003}
\begin{Theorem} \label{RITheoMGDecomp}
Let $(X_t)$ be a semimartingale with $X_0=1$, such that $(X_t)$ and $(X_{t-})$ are strictly positive. Then,  $(X_t)$ is a special semimartingale if and only if it admits a multiplicative decomposition 
\begin{equation}
X_t =  M_t L_t, \label{RIEqMultiDecomp}
\end{equation}
 where $(M_t)$ is a strictly positive and {\cadlag} $\M$-local martingale, $(L_t)$ is a positive, predictable process with locally finite variation, and $M_0=L_0=1$. When the decomposition exists, it is unique and is given by	
\begin{align*}
&M_t=\stochExp \left (\int_0^\cdot \frac{1}{X_{s-} + \Delta B_s} \dd N_s \right )_t,&
&L_t =\stochExp \left (-\int_0^\cdot \frac{1}{X_{s-} + \Delta B_s} \dd B_s \right )_t^{-1}&
\end{align*}
where $(B_t)$ and $(N_t)$ are the processes in its canonical additive decomposition.
\end{Theorem}
Comparing Eq.~\eqref{RIeqhtdecomp} and Eq.~\eqref{RIEqMultiDecomp} in Theorem~\ref{RITheoMGDecomp} we see that a number of conditions need to be satisfied for a short rate representation to be available. The pricing kernel $(h_t)$ has to satisfy the assumptions of Theorem~\ref{RITheoMGDecomp}, $(M_t)$ has to be a true martingale and act as a measure change, $(L_t)$ needs to have the specific form in \eqref{RIeqhtdecomp} and finally the resulting ZC bond prices need to be sufficiently differentiable. More precisely, following \citep[Theorem 5 and Proposition 12]{DoberleinSchweizer2001}, when \eqref{RIEqMultiDecomp} exists for a pricing kernel, one calls $A_t = L_t^{-1}$ an implied savings account. When in addition $A_t = 1 + \int_0 ^t \phi_s \dd s $ is satisfied, where 
$(\phi_t)$ is adapted and $\int_0^t \abs{\phi_s} \dd s \in L^1(\M)$, then the forward and short rates exist, that is, \eqref{RIeqhtdecomp} holds, and $(A_t)$ is termed a classical savings account.

Our first endeavour is to characterise the real-economy risk-neutral measure $\Q^{\suptext{R}}$. In the case that the short-rate process exists, we denote it by $(r^{\suptext{R}}_t)_{0\leq t}$. If in addition $(r^{\suptext{R}}_t)$ is absolutely integrable, then the discount factor $(D^R_t)$ exists, and we define
\begin{equation*}
D_t^{\suptext{R}} = \exp\left(-\int_0^t r^{\suptext{R}}_s \diff s\right). 
\end{equation*}
We introduce the process $(I^{\suptext{R}}_t )_{0 \leq t}$, given by
\begin{equation*}
I^{\suptext{R}}_t =  \int_0^{t} \frac{   b^{\suptext{R}}(s) } {    1+ b^{\suptext{R}}(s) (  A_{s-}^{\suptext{R}} -1 ) } \diff A_s ^{\suptext{R}},  
\end{equation*}
and note that  $\Delta I^{\suptext{R}}_t > -1$, for all $t\geq 0$.
Next, we denote by $\stochExp(\cdot)$ the stochastic exponential and define
\begin{equation} \label{RIEqxiR}
\xi^{\suptext{R}}_t=\stochExp \left (  I^{\suptext{R}} \right ) _t,
\end{equation} 
which is strictly positive for all $t\geq 0$.
\begin{Lemma} \label{RILemMChangeR}
Assume an RPKS. Then $(h_t^{\suptext{R}},\M)$ has a  savings account if and only if  $(\xi_t ^{\suptext{R}})$ in \eqref{RIEqxiR} is an $\M$-martingale. In this case, 
	$\xi_t ^{\suptext{R}}  =\dd \Q^{\suptext{R}}/\dd\M\lvert _{\Fc_t}$, and the savings account is classical with short rate process  $(r_t^{\suptext{R}})$ given by
	\begin{equation}\label{shortrate-rR}
r_t^{\suptext{R}} = -\frac{1}{{h_t^{\suptext{R}}}}\left(b_0'(t)  +b_1'(t) A_{t}^{\suptext{R}}\right).
	\end{equation}
\end{Lemma}
\begin{proof}  
	By the relation~\eqref{RIeqRealBond}, we obtain expression (\ref{shortrate-rR}). Now, $h_t^\suptext{R}= b_0(t)+b_1(t)A_t^\suptext{R}$ and Ito's formula shows
	\begin{equation} \label{RIEqRealRationalAdditiveDecomp}
	\dd h_t^{\suptext{R}} = -r_{t}^\suptext{R} h_{t}^\suptext{R} \dd t  + R(t) b^{\suptext{R}}(t) \dd A_t^{\suptext{R}}, \quad h_0^{\suptext{R}} = C_0  
	\end{equation}
	which exposes the unique additive decomposition of $(h_t^{\suptext{R}})$. One can now either calculate the dynamics of $\xi_t^\suptext{R} =h_t^\suptext{R} /D_t ^\suptext{R}$ or apply the formula for the multiplicative decomposition.
\end{proof}
We now examine the nominal market processes; this endeavour is slightly more elaborate. In the case that the nominal short-rate process $(r^{\suptext{N}}_t)_{0\leq t}$ and the associated discount factor $(D^N_t)$ exist, we write
\begin{equation*}
D_t^{\suptext{N}} = \exp\left(- \int_0^t r^{\suptext{N}}_s \diff s\right).
\end{equation*}
We furthermore define 
\begin{equation*}
I^{\suptext{S}}_t  = \int_0^t \frac{1}{s_{s-}} \diff s_s = \int_0^{t} \frac{ 1 } {  A^{\suptext{S}}_{s-}} \diff A_s ^{\suptext{S}},
\end{equation*}
and $m_t(T) = \EM_t[A_T^{\suptext{R}} A_T^{\suptext{S}}]$, which is differentiable in $T=t$ in the case that the nominal short rate exists in an RPKS. In the case that $\int_0 ^t \abs {\frac{m'_s(s)b_3(s)}{h_s^N}} \dd s < \infty$, we define the stochastic exponential
$
\xi_t^{\suptext{N}} = \stochExp\left(  \int_0 ^\cdot -\frac{m'_s(s)b_3(s)}{h_s^N} \dd s + I^{\suptext{R}}+ I^{\suptext{S}}  + [ I^{\suptext{R}}, I^{\suptext{S}}] \right)_t.
$
We note that $\nobreak{\Delta ( I^{\suptext{R}}_{ t} + I^{\suptext{S}}_{t}  + [ I^{\suptext{R}}, I^{\suptext{S}}]_ {t}  ) >-1}$ $\forall  t \geq  0  $, i.e. $\xi_t^{\suptext{N}} > 0$ for all $t\geq 0$. We define the ``pseudo short-rate'' 
\begin{equation}\label{pseudorate}
\fakeRate_t =-\frac{1}{h_t^{\suptext{N}}}\left( b_2'(t) A_t^{\suptext{S}} + b_3'(t) A_t^{\suptext{R}} A_t^{\suptext{S}}\right).
\end{equation}
\begin{Lemma}  \label{RILemMChangeN}
	Assume an RPKS and that $(r_t^\suptext{N})$ exists. Consider the ``pseudo short-rate'' (\ref{pseudorate}). Then $(h_t^{\suptext{N}},\M)$ has a classical savings account with short rate $(r_t^\suptext{N})$, given by 
	\begin{equation}\label{shortrate-rN}
	r_t^\suptext{N}  = r_t - \frac{m'_t(t) b_3(t)}{h_t^\suptext{N}},
	\end{equation}
	if and only if $\int_0 ^t \abs {\frac{m'_s(s)b_3(s)}{h_s^N}} \dd s < \infty$  and $(\xi_t^{\suptext{N}}) $ is an $\M$-martingale.
\end{Lemma}
\begin{proof}
	By \eqref{RIeqNomBondPrice} we get expression (\ref{shortrate-rN}).
	Since $h_t^N = b_2(t) A_t^{\suptext{S}} + b_3(t) A_t^{\suptext{R}} A_t^{\suptext{S}}$  Ito's formula shows that
	\begin{align}
	\dd h_t^\suptext{N} = {} &
	 - h_{t}^\suptext{N} \fakeRate_{t}  \dd t  + b_3(t) A_{t-}^\suptext{S} \dd A_t^\suptext{R}  \label{RIeqhNDecomp}
	+ (b_2(t) +b_3(t) A_{t-}^\suptext{R} ) \dd A_t^\suptext{S} + b_3(t) [A^\suptext{R},A^\suptext{S}]_t   \\
	= {} & \left ( - h_{t}^\suptext{N} r_{t}^\suptext{N}  - m_{t}'(t)b_3(t)  \right ) \dd t  + b_3(t) A_{t-}^\suptext{S} \dd A_t^{\suptext{R}}
	+ (b_2(t) +b_3(t) A_{t-}^\suptext{R} ) \dd A_t^\suptext{S} +  b_3(t) [A^\suptext{R},A^\suptext{S}]_t  \nonumber,
	\end{align}
	with $h^\suptext{N}_0=1$.
	We define $\xi_t^N= \frac{h_t^N}{D_t^N}$. An application of {Ito}'s quotient rule shows that
	\begin{align*}
	\dd \xi_t ^\suptext{N} =  \xi_t^\suptext{N}  \left (  -\frac{m'_t(t) b_3(t)}{h_t^\suptext{N}} \dd t   +\dd I^\suptext{R}_t   +\dd I^\suptext{S}_t +\dd[I^\suptext{R},I^\suptext{S}]_t \right ),
	\end{align*}
	i.e. $(\xi_t ^\suptext{N})$ is the stated stochastic exponential.
\end{proof}
Lemma~\ref{RILemMChangeN} is a weaker result than Lemma~\ref{RILemMChangeR}, since $(\xi_t^\suptext{N} )$ is not necessarily a local martingale. The result is still interesting because $(r_t^N)$ is the most tempting candidate for a short-rate process in a discount factor emerging from a nominal pricing kernel. Lemma~\ref{RILemMChangeN} shows that, in general, we cannot expect $(h_t^\suptext{N},\M)$ to have a classical savings account with short rate $(r_t^\suptext{N})$. This rules out though neither the existence of a savings account nor a classical savings account with a different ``short-rate''. To apply the theory we need $(h_t^N)$ to be a special semimartingale.
\begin{Lemma}
	Assume an RPKS and that $(h_t^{\suptext{N}})$ is a special semimartingale. Then the canonical additive decomposition is given by
	$h_t^{\suptext{N}} =  1+ B_t + N_t,$
	where
\begin{align}\label{RIEqhDAdditiveDecomp}
	\dd B_t ={} & h_t^\suptext{N} \fakeRate_{t} \dd t +  b_3(t)  \dd  \langle A ^{\suptext{R}}, A^{\suptext{S}} \rangle_t, &&  (B_0=0)  \\
	\dd N_t ={} &    b_3(t) A_{t-}^\suptext{S} \dd A_t^\suptext{R}+(b_2(t) +b_3(t) A_{t-}^\suptext{R} ) \dd A_t^\suptext{S}  +b_3(t) \diff \left(  [A^{\suptext{R}},A^{\suptext{S}}]_t- \langle A^{\suptext{R}}, A^{\suptext{S}} \rangle_t \right), \quad && (N_0=0). \notag
\end{align}
The multiplicative decomposition is $h_t^N =  M_t L_t$ where
\begin{align*}
	&M_t =\stochExp \left ( \int_0 ^\cdot  \frac{1}{h_{s-}^\suptext{N} + \Delta B_s} \dd N_s \right )_t,&
&L_t =\stochExp\left ( -\int_0 ^\cdot  \frac{1}{h_{s-}^\suptext{N} + \Delta B_s} \dd B_s \right )_t ^{-1},&
\end{align*}
moreover, a savings account exists if and only if $(M_t)$ is a martingale.
\end{Lemma}
\begin{proof}
	The dynamics of $(h_t^\suptext{N})$ was derived in \eqref{RIeqhNDecomp}. By \citep[I Theorem 4.23]{JacodShiryaev2003}, 
	$\int_0^t b_3(s) \diff [A ^{\suptext{R}},A^{\suptext{S}}]_s $
	has locally integrable variation, and therefore it has a predictable compensator $$\int_0^t b_3(s) \diffIN \langle A^{\suptext{R}},A^{\suptext{S}}\rangle_s,$$ see \citep[I Theorem 3.18]{JacodShiryaev2003}. 
\end{proof}
When there are no simultanous jumps in $(A^\suptext{R}_t)$ and $(A^\suptext{S}_t)$ the situation is simpler:
\begin{Corollary}
	Assume an RPKS, and that $\Delta A_t^\suptext{R} \Delta A_t^\suptext{S} = 0$ a.s. Then $(h_t^{\suptext{N}})$ is a special semimartingale and the canonical decomposition is given by	$h_t^{\suptext{N}} =  1+ B_t + N_t,$
	where
\begin{align*}
	\dd B_t ={} & h_t^\suptext{N} \fakeRate_{t} \dd t+  b_3(t)  \dd [A^{\suptext{R}},A^{\suptext{S}}]_t, && (B_0=0) \\
	\dd N_t ={} &    b_3(t) A_{t-}^\suptext{S} \dd A_t^\suptext{R}+(b_2(t) +b_3(t) A_{t-}^\suptext{R} ) \dd A_t^\suptext{S},  \quad && (N_0=0)
\end{align*}
	and $h_t^\suptext{N} =  M_t L_t$ where
\begin{align*}
&M_t =\stochExp \left (I^\suptext{R} + I^\suptext{S} \right )_t,&
&L_t =\stochExp\left( -\int_0^\cdot  h_s^\suptext{N} \fakeRate_{s} \dd s +[I^\suptext{R},I^\suptext{S}]\right)_t ^{-1},&
\end{align*}
and a savings account exists if and only if $(M_t)$ is a martingale.  If additionally $[A^{\suptext{R}} , A^{\suptext{S}}]_t$  is absolutely continuous,  write $[A^{\suptext{R}} , A^{\suptext{S}}]_t = \int_0 ^t a_s \dd s$ and define $\lambda_t = a_t b_3(t)/h_t^N$ for $t\geq 0$. Then,
\begin{equation*}
L_t =\exp \left ( \int_0^\cdot (h_s^\suptext{N} \fakeRate_{s} - \lambda_s) \dd s \right )_t
\end{equation*}
and, if $\int_0^t  \exp(-\int_0 ^s (h_u^N \fakeRate_{u} - \lambda_u) \dd u )  \abs{h_s^N \fakeRate_{s} - \lambda_s} \dd s \in L^1(\M)$, then a classical savings account exists.
\end{Corollary}
\begin{proof}
	By assumption, we have that $[A^{\suptext{R}},A^{\suptext{S}}]_t=\langle (A^{\suptext{R}})^c, (A^{\suptext{S}})^c \rangle_t$ and thus 
	the decomposition above has a predictable bounded variation part. 
	By change of variables, $L_t^{-1} = 1+\int_0 ^t \frac{\partial}{\partial s}\exp(-\int_0 ^s (h_u^N \fakeRate_{u} - \lambda_u) \dd u ) \dd s$. That is, by \citep[Proposition 12]{DoberleinSchweizer2001}, integrability implies existence of a classical savings account.
\end{proof}
A simple example where the first condition is satisfied is if $(A_t^{\suptext{R}})$ or $(A_t^{\suptext{S}})$ is continuous. If both $(A_t^{\suptext{R}})$ and $(A_t^{\suptext{S}})$ are Ito processes then $[A^{\suptext{R}} , A^{\suptext{S}}]_t$ is absolutely continuous.

Finally we present a lemma giving a condition in an RPKS to check whether $(h_t^{\suptext{N}})$ is special, this will be particularly simple to check in the following setting.
\begin{Lemma} \label{RILemSpecialSquare}
	Assume an RPKS and that  $(A_t^{\suptext{R}})$ and $(A_t^{\suptext{S}})$  are locally square-integrable. Then $(h_t^{\suptext{N}})$ is a special semimartingale.
\end{Lemma}
\begin{proof}
	By \cite[I Proposition 4.50]{JacodShiryaev2003}, the process $([A^{\suptext{R}},A^{\suptext{S}}]_t)$ has locally integrable variation and its compensator $(\langle A^{\suptext{R}}, A^{\suptext{S}}\rangle_t)$ exists. Therefore, the decomposition in \eqref{RIEqhDAdditiveDecomp} is warranted.
\end{proof}
\section{Construction of the exponential-rational class}  \label{RISecRatExpCons}
Our next goal is to derive explicit price formulae for financial derivatives based on the ZC and the YoY swap rates, and for the so-called limited price-index (LPI) swap. For its flexibility, tractability and good calibration properties, we choose to work with a sub-class among the rational pricing kernel systems, namely the exponential-rational pricing kernels. We next construct this class. 
\begin{Definition}[Exponential-rational pricing kernels] \label{RIDefERPK}
Assume an RPKS and let  $(X_t)_{0\leq t}$ be a d-dimensional stochastic process.  Assume that  $(A^{\suptext{R}}_t)_{0\leq t}$ and $(A^{\suptext{S}}_t)_{0\leq t}$ in Definition~\ref{RIExModelCons} are on the form  $A^{\suptext{R}}_t=\exp(\innerp {w_{\subtext{R}}}{X_t})$ and $A^{\suptext{S}}_t= \exp(\innerp{w_{\subtext{S}}}{X_t})$. We call this class the exponential-rational pricing kernel models. If $(X_t)$ is an additive process, we call this class the  additive exponential-rational pricing kernel models.
\end{Definition}
Remember that an RPKS requires that $(A_t^{\suptext{R}})$ and $(A_t^{\suptext{S}})$ are martingales, that $(A_t^RA_t^S)$ has finite expectation for all $t\geq 0$ and that $A_0^{\suptext{S}}=A_0^{\suptext{R}}=1$, it is implicit that $w_{\subtext{R}},w_{\subtext{S}}$ and $(X_t)$ in Definition~\ref{RIDefERPK} are chosen such that this is satisfied. 
\begin{Definition}[Additive process]  \label{RIDefAdd}
 Let  $(X_t)_{0\leq t  }$ be a d-dimensional stochastic process. Following  \citep[Definition 1.6]{Sato1999}, we say $(X_t)$ is additive if it has a.s. {\cadlag} paths, $X_0=0$, and
 \begin{enumerate}
 \item Independent increments: for any $n\geq 1$, $0\leq t_0 < t_1 < \dots<t_{n-1} <  t_n$,  the random variables $X_{t_0}, X_{t_1}- X_{t_0},\dots, X_{t_n}-X_{t_{n-1}}$ are independent,
 \item Stochastic continuity: for any $t\geq 0$ and  $\epsilon >0 $,  $\lim_{s \to t } \M ( \abs{ X_t - X_s} > \epsilon ) =0$.
 \end{enumerate}  
\end{Definition}
In the remainder of the paper, the derivations will be based on exponential-rational pricing kernel models. The additive exponential-rational pricing kernel models will be given particular attention, so we  now recall some facts about additive processes and provide some examples. The independent increments property gives a {\Levy}-{Khintchine} representation
\begin{align*}
\EM \left[ \e^{i   \innerp{z } {  X_t}  }\right] ={} & \e^{  i \psi_t(z)},  \\
\psi_t(z) = {} & i \innerp{z}{ \mu_t}  - 1/2 \innerp{z}{ \Sigma_t  z } +\int_{\R^d  }  \left(  \e^{i \innerp{z}{x}}-1 - i \innerp{z}{x}  \1{ \norm{x} \leq 1 }   \right) \intspace \nu_t(\dd x) ,
\end{align*}
where the {\Levy}-Khintchine triplet ($\mu_t$,$\Sigma_t$,$\nu_t$) is unique and satisfies a number of conditions (see \citep[Chapter 2]{Sato1999}). The {\Levy}-Khintchine triplet also  determines the sample path properties of $(X_t)_{0\leq t}$ by the {\Levy}-Ito decomposition  (see  \citep{Sato1999}[Chapter 4]). Next, we consider  examples of how additive processes can be obtained. 
\begin{Example}[Time change]  \label{RIExTimeChange}
Let $(\tau_t)_{0\leq t}$ be a continuous, increasing process with $\tau_0=0$. For $t\geq 0$, define pathwise $X_t = L_{\tau_t}$  for $(L_t)_{0\leq t}$ a {\Levy} process. Then $(X_t)_{0 \leq t}$ inherits the independent increments of $(L_t)$ and is therefore additive. A particular simple case is obtained by letting $\tau_t=\tau(t)$ be deterministic.  Then, we may write 
$
\EM \left[ \exp ( i  z   X_t ) \right] =  \exp ( \tau(t)  \psi^{\suptext{L}}(z)), 
$
where $ \psi^{\suptext{L}}(z)$ is the characteristic exponent of $(L_t)$ in $t=1$.
\end{Example}

\begin{Example}[Stacking independent additive processes]  \label{RIExStackingAdditive}
Let $(X_t^1),\dots,(X_t^N)$ be one-dimensional additive processes with characteristic exponents $ \psi^1_t ,\dots,\psi^N_t$. Then $({X}_t) = ( X^1_t,\dots,X^N_t)$  is an N-dimensional additive process with characteristic exponent 
$\psi_t({z}) = \sum_{i=1}^N \psi^i_t(z_i), $
where ${z}= (z_1,\dots,z_N)$. Furthermore, for $w\in \R^N$ $(\innerp{w}{X_t})_{0 \leq t}$ is a one-dimensional additive process with characteristic exponent $z \mapsto \psi_t(z w)$.
\end{Example}

A fact about  additive processes is that they are convenient to construct martingales. Define
\begin{equation} \label{RIeqExpMomCond}
\assumExp(X) = \left\{ z \in \mathbb{C}^d :  \int_{\{x \in \R^d \, : \,  \norm{x}> 1 \} }    \exp(\innerp{\Re z}{x})   \intspace \nu_t(\dd x) < \infty, \quad \forall t\geq 0 \right\}, 
\end{equation}
then for $z \in \assumExp(X)$ it holds that $\EM [ \abs{\e^{\innerp{z}{X_t}}}]<\infty$ and the Laplace exponent $\kappa_t(z) = \psi_t (  -i  z )$ is well-defined, see \cite{Sato1999}[Theorem 25.17].
It follows from the independent increments property that for $w\in\assumExp(X) $, one has
\begin{equation} \label{RIeqFundamentalMartingale}
\frac{\e^{\innerp{w}{X_t}} }{\EM\left[\e^{\innerp{w}{X_t}}\right] } = \e^{\innerp{w}{X_t}  - \kappa_t(w) } 
\end{equation} 
is a martingale. We can build exponential martingales by taking an additive process and let the drift absorb the mean in~\eqref{RIeqFundamentalMartingale}. This produces the condition that, if
\begin{equation}\label{RIeqExpMGCond}
\innerp{w}{\mu_t} =  -\frac{1}{2} \innerp{w}{ \Sigma_t  w }- \int_{\R^d}  \left(  \exp(\innerp{w}{x})-1 -  \innerp{w}{x} \1{ \norm{x} \leq 1  } \right)\intspace \nu_t(\dd x),  
\end{equation}
then $(\e^{\innerp{w}{X_t}})_{0\leq t }$ is a martingale. Eq.~\eqref{RIeqExpMomCond} is useful for additive-exponential rational models. Recall Definition~\ref{RIExModelCons}, Eq.~\eqref{RIeqExpMomCond}  shows that  $(A_t^R A_t^S)$ has finite expectation for all $t\geq 0$ if $w_{\subtext{S}}+w_{\subtext{R}} \in \assumExp(X)$. Similarly, recalling Lemma~\ref{RILemSpecialSquare}, $(A_t^R)$ and $(A_t^S)$ are locally square-integrable if $2 w_{\subtext{S}}\in \assumExp(X) $ and $2 w_{\subtext{R}} \in \assumExp(X)$.

We proceed to calculate a number of expressions needed in both the previous and next sections. First,   for $w\in \assumExp(X)$ and $0 \leq t\leq T$,
\begin{equation*}
\EM_t\left[ \exp ( \innerp{w} {X_T})\right] =  \exp(\innerp{w}{X_t}) \exp ( \kappa_{tT} (  w) ),
\end{equation*}
where we define  the forward Laplace exponent $\kappa_{tT}(\cdot) = \kappa_T (\cdot ) -\kappa_t(\cdot) $. For the  YoY swap~\eqref{RIeqYoYVal} we also need
\begin{align*}
\Cov^{\M}_t\left[ \exp ( \innerp{w_1} {X_T}) , \exp( \innerp{w_2} {X_T} )\right] 
									= {} &   \e^{\innerp{w_1+w_2}{X_t}}  \left (  \e^{\kappa_{tT} (w_1+w_2)   } - \e^{\kappa_{tT}(w_1) +\kappa_{tT}(w_2)  }   \right )  \\
									= {} &\e^{\innerp{w_1+w_2}{X_t}}    \e^{\kappa_{tT} (w_1+w_2)  }  \left ( 1  - \e^{\kappa_{tT}(w_1) +\kappa_{tT}(w_2)- \kappa_{tT} (w_1+w_2)   }   \right ), 
\end{align*}
where $w_1,w_2, (w_1+w_2) \in \assumExp(X)$ is assumed.
 We notice that the sign, and to some extent the magnitude of the covariance, depends on the non-linearity of $z\mapsto \kappa_{tT}(z)$.
 For the subsequent derivation of Fourier-inversion formulae, we will also need the multiperiod generalized characteristic function. 
 \begin{Lemma} \label{RILemMPCF}
Let $(X_t)_{0\leq t}$ be an additive process.  Assume $t \leq T_0\leq T_1 \leq \dots \leq T_N$, set $u = ( u_1, \dots, u_N)$ and define
 \begin{equation*}
  q_t(u) =   \EM_t \left[  \exp \left(   \, \sum_{i=1}^N  u_i \innerp{w_i}{X_{T_i}}  \, \right)  \right].
 \end{equation*}
 Set  $z_N =  u_N w_N$ and  $z_{i-1} = z_i +  u_{i-1} w_{i-1}$ for $i=2,\dots,N$. Assume that $z_i \in \assumExp(X)$ for $i=1,\dots,N $.
Then,
\begin{equation} \label{RIeqMPCF}
q_{t}(u)  =\exp( \innerp{z_1({u})}{ X_{t}} ) \exp \left ( \,  \sum_{i=1}^N  \kappa_{ T_{i-1} T_i } ( z_i({u} ) ) \, \right  ), 
\end{equation}
which is well-defined and finite.
 \end{Lemma}
 \begin{proof}
 The statement follows from iterated expectation and the independent increments property.
 \end{proof}

We next build on the ideas of the Examples~\ref{RIExTimeChange} and \ref{RIExStackingAdditive} with two concrete specifications for the  additive exponential-rational pricing kernel models. Very similar models will be calibrated in Section~\ref{RISecCaliExam}.

\begin{Specification}[Time-changed {\Levy}]    \label{RIExTimeChangeLevy}  
Let $(L_t)_{0\leq t }$ be a  {\Levy} process satisfying the condition $\Var^{\M}(L_1)=1$. Let  $(t_1,a_1) , (t_2,a_2) , \dots , (t_n,a_n)$  be given points such that both coordinates are increasing and let $\tau(t)$ be a continuous, non-decreasing  interpolation. Then   $t \mapsto \Var^{\M}(L_{\tau(t)}) = \tau(t)$  interpolates  the given points. This property can be utilised for, e.g., fitting a term-structure of at-the-money implied volatilities. Building on this motivation, we let 
\begin{equation*}
X_t = \left(L^{\suptext{R}}_{\tau^{\suptext{R}}(t)} + \mu^{\suptext{R}}(t) ,L^{\suptext{S}}_{\tau^{\suptext{S}}(t)}+ \mu^{\suptext{S}}(t)\right).
\end{equation*}
Let $(t \tilde{\mu}^{\suptext{R}},t  (\sigma^{\suptext{R}})^2,  t \nu^{\suptext{R}}) $ and $(t \tilde{\mu}^{\suptext{S}}, t (\sigma^2)^{\suptext{S}},t  \nu^{\suptext{S}})$ denote the {\Levy}-Khintchine triplets of $(L_t^{\suptext{R}})$ and $(L_t ^{\suptext{S}})$.
Then $(X_t)$ has the {\Levy}-Khintchine triplet
\begin{align*}
 & \mu_t = \left(   \mu^{\suptext{R}}(t) +\tau^{\suptext{R}}(t) \tilde{\mu}^{\suptext{R}} ,   \mu^{\suptext{S}}(t) + \tau^{\suptext{S}}(t) \tilde{\mu}^{\suptext{S}} \right),  \quad  \Sigma_t =  \begin{pmatrix} (\sigma^{\suptext{R}})^2   \tau^{\suptext{R}}(t)   & 0 \\ 0 &  (\sigma^{\suptext{S}})^2    \tau^{\suptext{S}}(t)  \end{pmatrix},  \\
& \nu_t(B)  =   \tau^{\suptext{R}}(t)   \nu^{\suptext{R}} ( B_1 )  +  \tau^{\suptext{S}}(t)   \nu^{\suptext{S}} ( B_2 )   
\end{align*}
where $B_1 = \{ x \in \R: (x,0) \in B \}$ and $B_2 =  \{ x\in \R : (0,x) \in B \}$. 
By choosing the drifts $\mu^{\suptext{R}}(t)$ and $\mu^{\suptext{S}}(t)$ according to~\eqref{RIeqExpMGCond} we can turn $A^{\suptext{R}}_t=\exp(\innerp {w_{\subtext{R}}}{X_t})$ and $A^{\suptext{S}}_t= \exp(\innerp{w_{\subtext{S}}}{X_t})$ into martingales. This construction  generalises to higher dimensions in a straightforward way.
\end{Specification}

\begin{Specification}[Time-changed Wiener process] \label{RIExTimeChangeWP}  
As a special case of Example~\ref{RIExTimeChangeLevy}, we consider  time-changing independent Wiener processes. 
We set 
\begin{align*}
X_t = \left(  W^{\suptext{R}}_{\tau^{\suptext{R}}(t)} + \mu^{\suptext{R}}(t), W^{\suptext{S}}_{\tau^{\suptext{S}}(t)} + \mu^{\suptext{S}}(t) \right) .
\end{align*}
This corresponds to having the {\Levy}-Khintchine triplet given by 
\begin{align*}
\mu_t = {} & \left( \mu^{\suptext{R}}(t), \mu^{\suptext{S}} (t) \right) , &  \; \nu_t ={} &  0 , & \; \Sigma_t =  {} &   \begin{pmatrix}    \tau^{\suptext{R}}(t)   & 0 \\ 0 &     \tau^{\suptext{S}}(t)  \end{pmatrix}.
\end{align*}
We can choose $(A^{\suptext{R}}_t)$ and $(A^{\suptext{S}}_t)$ to be martingales as in Specification~\ref{RIExTimeChangeLevy}.
\end{Specification}
\subsection{Option pricing}  \label{RISecOptPric}
By use of the exponential-rational pricing kernel models, tractable expressions can be derived for inflation-linked derivatives, such as the YoY floor and the ZC floor. Under the stronger assumption of  additive exponential-rational pricing kernel models, we can find a similarly tractable formula for the LPI swap.
\subsubsection{Year-on-Year floors} \label{RISecYoYF}
The payoff of the YoY floor can be written in terms of a series of floorlets. A floorlet has payoff function $( K- C_{T_i}/C_{T_{i-1}})^+ $ paid at time $T$, typically $\quoteStrike$ is quoted with $K=1+\quoteStrike$. In practice it is often observed that $T> T_i$  to ensure that there is a reliable observation of CPI available at maturity.  We want our framework to be able to accommodate this feature. The next theorem is a pricing formula for the YoY floorlet.
\begin{Theorem} \label{RITheoYoYFourier} 
Assume an exponential-rational pricing kernel model. Let $Y_1=c_1 +  \innerp{ w_{\subtext{S}}}{ X_{T_{i-1} } }-\innerp{  w_{\subtext{S}}}{ X_{T_i}  }$, $Y_2=\innerp{w_{\subtext{S}}}{ X_{T}}$ and $Y_3 =c_3 + \innerp{w_{\subtext{R}}}{X_T}$
where $c_1=\ln[S(T_{i-1})/(K S(T_i))]$, $c_3=\ln[b^{\suptext{R}}(T)/(1-b^{\suptext{R}}(T))]$,   
and $q_t(z)= \EM_t[\e^{ \innerp{z}{  (Y_1,Y_2, Y_3)}}]$. Let $R>0$ and assume that 
\begin{equation}  \label{RITheoYoYFourierEqIntCond}  
 q_t(-R,1,1)+q_t(-R,1,0)<\infty. 
 \end{equation}
Let  $t\leq T$ and consider
\begin{equation*}
V^{\suptext{YoYFl}}_{t} =   \frac{1  } {h_t^{\suptext{N}} }  \EM_t \left [ h_T^{\suptext{N}} \left ( K-  \frac{C_{T_i}}{C_{T_{i-1}}}   \right )^+ \right ],  
\end{equation*}
by Eq.~\eqref{RIeqPricingFormulaM}, the price of the YoY floorlet. Then we have: 
\begin{equation} \label{RIeqYoYPriceExp} 
V^{\suptext{YoYFl}}_{t}    =  \frac{c_0 K }{\pi  h_t^{\suptext{N}} } \int_{\R^+} \Re \frac{\vartheta_t  ( u )  }{(R+i u) (1+R+iu)}  \diff  u,  
\end{equation}
where $c_0  = R(T) (1-b^{\suptext{R}}(T))S(T)$ and $\vartheta_t ( u ) =q_t(-(R+iu),1,0)+ q_t(-(R+iu),1,1)$. If $T_i \leq t \leq T $, then  $V^{\suptext{YoYFl}}_t  = (K-  C_{T_i}/C_{T_{i-1}})^+ P_{tT} ^{\suptext{N}}$.
\end{Theorem}
\begin{proof}
By the pricing formula~\eqref{RIeqPricingFormulaM} the price at any time $t < T_{i}$   is
\begin{align*}
V^{\suptext{YoYFl}}_{t}  = &{}     \frac{1  } {h_t^{\suptext{N}} }  \EM_t \left [ h_T^{\suptext{N}} \left ( K- \frac{C_{T_i}}{C_{T_{i-1}}}   \right )^+ \right ]
	   = \frac{  c_0 K }{h_t^{\suptext{N}} } \EM_t \left[(1+\e^{Y_3}) \, \e^{Y_2}  \,  (1- \e^{Y_1})^+ \right].
\end{align*}
 We can directly apply Lemma~\ref{RILemKey}, found in the appendix, to get~\eqref{RIeqYoYPriceExp}.  Note that if $T_{i-1}\leq t<T_i$  then a part of $Y_1$ is measurable. The formula for $ T_i \leq t \leq T$ follows by observing that the payoff function is $T_i$-measurable and recalling Eq.~\eqref{RIeqNomBondPriceGen} combined with the relation~\eqref{RIeqPricingFormulaM}.
\end{proof}
A point to note about~\eqref{RIeqYoYPriceExp} is the quadratic convergence of the numerator in the integral, which makes the formula particularly tractable. 
In the case that we use an additive exponential-rational pricing kernel model, $q_t(z)$ follows from an   application of Lemma~\ref{RILemMPCF}, and the integrability~\eqref{RITheoYoYFourierEqIntCond}  is satisfied if $-R w_{\subtext{S}} \in \assumExp(X)$.
Note that the difference between a YoY caplet and floorlet is a YoY swaplet, potentially with time-lag. The price is 
\begin{equation} \label{RIEqYoYCapletFloorletDiff}
V_t^{\suptext{YoYSL}}= \frac{1}{h_t^{\suptext{N}}} \EM_t \left [ h_{T}^{\suptext{N}}  \left  (\frac{C_{T_i}}{C_{T_{i-1}}} -K \right )  \right ] =     \frac{  c_0 K}{h_t^{\suptext{N}} } \left (  q_t( 1,1,0)+q_t( 1,1,1) \right )-KP_{tT}^{\suptext{N}}, 
\end{equation}
where $q_t(z)$ and $c_0$ are given in Theorem~\ref{RITheoYoYFourier}.  This can be used in conjunction with Theorem~\ref{RITheoYoYFourier} to get the price of a caplet or used directly to price the time-lagged swaplet. 
 \subsubsection{Zero-Coupon floors} \label{RISecZCF} 
Next we focus on the pricing of the ZC floor, which together with the ZC cap and  YoY caps and floors,  are the most liquidly traded inflation-linked derivatives. The structure of this section closely follows   the previous one, since the calculations   are   similar.  The payoff at time $T$ of the ZC floor can  be written in the form $( K-C_{T_i}/C_0)^+$ with $T \geq T_i$ akin to the YoY floor. Typically, the strike quoted is $\quoteStrike$, where $K=(1+\quoteStrike)^T$.
\begin{Theorem} \label{RITheoZCFourier}
Assume an exponential-rational pricing kernel model. Let $Y_1= c_4 +  \innerp{ w_{\subtext{S}}}{ X_{T_i } }$, $Y_2= \innerp{  w_{\subtext{S}}}{ X_{T}  }$ and   $Y_3=c_3 +\innerp{  w_{\subtext{R}}}{ X_{T} }$, where
\begin{equation*}
   c_4={}   \ln  \left(\frac{1   }{ K S(T_i) }\right)        , \enskip  c_3= {}   \ln \left(\frac{b^{\suptext{R}}(T)}{1-b^{\suptext{R}}(T)}\right),   
\end{equation*}
and $q_t(z)= \EM_t [\e^{ \innerp{z}{  (Y_1,Y_2, Y_3)}}]$. 
Assume $R>0$ and
\begin{equation} \label{RITheoZCFourierEqAss}
q_t( -R,1,0)+q_t(-R,1,1) < \infty  . 
\end{equation}
Consider  $T_0 \leq t\leq T_{i}\leq T$ and let
\begin{equation*}
V^{\suptext{ZCF}}_{t}  =   \frac{1  } {h_t^{\suptext{N}} }  \EM_t \left [ h_T^{\suptext{N}} \left (  K -\frac{C_{T_i}}{C_{T_0}}  \right )^+  \right ]  
\end{equation*}
 be the price at time $t$ of a ZC floor. Then we have 
\begin{equation*}
V^{\suptext{ZCF}}_{t}    =  \frac{c_0 K }{\pi  h_t^{\suptext{N}} } \int_{\R^+} \Re \frac{\vartheta_t  ( u )  }{(R+i u) (1+R+iu)}  \diff u, 
\end{equation*}
where $c_0=R(T) \left(1-b^{\suptext{R}}(T)\right) S(T)$ and
$\vartheta_t  ( u ) =q_t\left(-(R+iu),1,0\right)+ q_t\left(-(R+iu),1,1\right)$. If $T_i \leq t \leq T $, then  $V^{\suptext{ZCF}}_t  = \left( K- C_{T_i}/C_{T_0}   \right)^+ P_{tT} ^{\suptext{N}}$. 
\end{Theorem}
\begin{proof}
Exactly as for the YoY case, see Theorem~\ref{RITheoYoYFourier}.
\end{proof}

If we assume the additive exponential-rational pricing kernel,  $q_t(z)$ follows directly from \eqref{RIeqMPCF}, and the assumption~\eqref{RITheoZCFourierEqAss} is satisfied if  $-R w_{\subtext{S}} \in \assumExp(X)$. Analogous to the YoY cap we  have the time-lagged ZC swap has price
\begin{equation} \label{RIEqValDiffZCCapFloor}
V_t^{ZCS}  = \frac{1}{h_t^{\suptext{N}}} \EM_t \left [ h_T^{\suptext{N}} \left ( \frac{C_{T_i}}{C_{T_0}}-K \right )  \right ]= \frac{c_0}{h_t^{\suptext{N}}} \left ( q_t(1,1,0)+ q_t(1,1,1) \right ) - K P_{tT}^{\suptext{N}},  
\end{equation} 
where $c_0$ and $q_t(z)$ are given in Theorem~\ref{RITheoZCFourier}. Which can also be used to obtain the price of ZC caps. 
\subsubsection{Limited price index swap} \label{RISecLPISwap}
The tractability of the model specification we have used so far allows to find semi-closed-form price formulae for the  exotic limited price index (LPI) swap. This, contrary to the previous theorems, does rely on the assumption that the driving stochastic process $(X_t)_{0\leq t}$ is additive.
The LPI is defined by 
\begin{equation*}
C^{\suptext{LPI}} _ {T_k}  = C^{\suptext{LPI}} _ {T_{k-1}}   \, \mathrm{mid} \left( 1+K_f , \frac{C_{T_k}}{C_{T_{k-1}}} , 1+K_c\right) ,
\end{equation*}
where $k=1,\dots,N$ and $T_k$ is a periodic fixed date, typically yearly. The contracts have maturities up to 30 years. Similar to the ZC swap, the LPI swap  has payoff  $C^{\suptext{LPI}} _{T_N} -K$. We will consider the payoff to be settled at the fixed time $T\geq T_N$. The pricing relation~\eqref{RIeqPricingFormulaM} gives the swap price at time $t$:
\begin{equation*}
V^{\suptext{LPIS}}_t  =\frac{1}{h_t^{\suptext{N}}} \EM_t\left[ h_{T}^{\suptext{N}}  C^{\suptext{LPI}} _{T_N} -K\right] = P_{t T}^{\suptext{LPI}}  - K P_{tT}^{\suptext{N}}
\end{equation*}
where $P_{t T}^{\suptext{LPI}} =\EM_t [ h_T^{\suptext{N}} C^{\suptext{LPI}}_{T_N}  ]/h_t^{\suptext{N}} $ is the price process of the  LPI-linked ZC bond. We therefore need to derive the price at time $t\leq T$ of the LPI-linked  ZC bond.
\begin{Theorem} \label{RITheoLPIBondFourier}
Assume an additive exponential-rational pricing kernel model.  Assume, without loss of generality, that $T_0 \leq t < T_1< T_2 <\dots < T_N \leq T$.
Let, for $k=1,\dots, N$, 
\begin{align*}
q^{1k} ( z_1,z_2) = {} & \EM \left[  \exp \left  ( (z_1+z_2) \innerp{w_{\subtext{S}}}{X_{T_k  } -X_{T_{k-1} \vee t  }}  \right )\right],   \\
q^{2k} ( z_1,z_2) = {} & \EM  \left[  \exp \left (  z_1 \innerp{w_{\subtext{S}}}{X_{T_k  } -X_{T_{k-1} \vee t    }  } + z_2  \innerp{w_{\subtext{R}}+w_{\subtext{S}}}{X_{T_k  } -X_{T_{k-1} \vee t}  } \right  )\right] ,  
\end{align*}
and $R_k>0$  be such that 
\begin{equation*}
\sum_{k=1}^N \left (q^{1k} ( -R_k,1)+q^{2k} ( -R_k,1) \right ) < \infty. 
\end{equation*}
Then,
\begin{equation*}
P_{tT}^{\suptext{LPI}}  = \tfrac{ 1}{h_{t}^{\suptext{N}}}  C^{\suptext{LPI}}_ {{T}_0} \left (   c_0 V_t^{11} \prod_{k=2}^N V^{1k} + c_5 A_t^{\suptext{R}}  V_t^{21} \prod_{k=2}^N V^{2k}  \right )   A_t^{\suptext{S}} 
\end{equation*}
where 
\begin{align*}
c_0 = {} &   R(T) \left(1-b^{\suptext{R}}(T)\right) S(T)  ,  & c_5 ={} & R(T) b^{\suptext{R}}(T) S(T)  \exp\left( \kappa_{ T_N T } (   w_{\subtext{R}}+w_{\subtext{S}}  )\right).
\end{align*}
Furthermore, for $j=1,2$
\begin{align*}
V_t ^{j1}= {} & \beta_c \exp \left (  \kappa_{ t T_1   } (   \omega^j )  \right ) +   \frac{1}{\pi} \int_{\R ^+ }  \Re \frac{ \beta_c \, \vartheta^{j1}(\alpha^{1c}_t,u) + \beta_f \, \vartheta^{j1}(\alpha^{1f}_t,u) }{(R+iu) (1+R+iu) }  \diff u,  \\
V ^{jk}= {} & \beta_c \exp \left (  \kappa_{  T_{k-1} T_k   } (  \omega^j )  \right ) +   \frac{1}{\pi} \int_{\R ^+ }  \Re \frac{ \beta_c \, \vartheta^{jk}(\alpha^{kc},u) + \beta_f \, \vartheta^{jk}(\alpha^{kf},u) }{(R+iu) (1+R+iu) }  \diff u,  
\end{align*}
where $\vartheta^{jk} (\alpha, u) =     \alpha^{-(R+iu  ) } q^{jk}   ( -(R+iu) ,1  )$ ,  and
\begin{align*}
 \alpha^{1c}_t = {} &   \beta_c^{-1} \frac{  S(T_1) }{  S( T_0 )} \frac{ A^{\suptext{S}}_{T_0 } }{ A_t^{\suptext{S}}    }  ,   &   \alpha^{1f}_{t} ={} & \beta_f^{-1} \frac{  S(T_1) }{  S( T_0 )} \frac{ A^{\suptext{S}}_{T_0 } }{ A_t^{\suptext{S}}    } ,   \\
  \alpha^{kc} = {} &   \beta_c^{-1} \frac{ S(T_{k-1}  )}{S(T_{k})} , &    \alpha^{kf} ={} &  \beta_f^{-1}\frac{ S(T_{k-1}  )}{S(T_{k})} ,  \quad \text{for } k=2,\dots,N, \\
 \beta_c ={} & (1+K_c) ,    &  \beta_f ={} & (1+K_f), \\
 \omega^1 ={} &  w^{\suptext{S}} , & \omega^2 = {} & w^{\suptext{S}}+w^{\suptext{R}} .
\end{align*}
\end{Theorem} 
\begin{proof}
First  we  write 
\begin{align*}
 C^{\suptext{LPI}}_ {T_N} = {} & C^{\suptext{LPI}}_ {{T}_0}\prod_{k=1}^N  Z_{T_k},   \\
 Z_{T_k} ={} &     (1+K_c )  -\left (1+ K_c - \frac{C_{T_{k}}}{C_{T_{k-1}}}  \right )^{+}+ \left ( 1+ K_f -  \frac{C_{T_{k}}}{C_{T_{k-1}}} \right )^+ .
\end{align*}
Note that $Z_{T_k} $ is $\Fc_{T_k}$-measurable and independent of $\Fc_{T_{k-1}  }$. 
Using the tower property and the independent increments property we have: 
\begin{align}
P_{t{T_N}}^{\suptext{LPI}}   = {} &   C^{\suptext{LPI}}_ {T_0} \frac{1 }{h_{t}^{\suptext{N}}} \EM_t  \left  [ h_{T} ^{\suptext{N}}  \prod_{k=1}^N  Z_{T_k}  \right ]  \notag \\
							  = {} & C^{\suptext{LPI}}_ {T_0} \frac{c_0 }{h_{t}^{\suptext{N}}}   \EM_t  \left  [  A_{T_{N-1}}^{\suptext{S}}  \prod_{k=1}^{N-1}  Z_{T_k}   \right   ]  \EM  \left [  \frac{ A_{T_N}^{\suptext{S}} }{A_{T_{N-1}}^{\suptext{S}} } Z_{T_N}  \right  ]     \notag \\
							   &\hspace{4.5cm}+C^{\suptext{LPI}}_ {T_0} \frac{c_5 }{h_{t}^{\suptext{N}}}    \EM_t \left  [ A_{T_{N-1}}^{\suptext{R}}   A_{T_{N-1}}^{\suptext{S}}  \prod_{k=1}^{N-1}  Z_{T_k}    \right  ]  \EM  \left  [  \frac{ A_{T_N}^{\suptext{R}}  A_{T_N}^{\suptext{S}} }{ A_{T_{N-1}}^{\suptext{R}} A_{T_{N-1}}^{\suptext{S}}} Z_{T_N} \right  ]      \notag  \\
\begin{split} \label{RITheoLPIEqExp}  
=  &  {}   C^{\suptext{LPI}}_ {T_0} \frac{c_0 }{h_{t}^{\suptext{N}}} A^{\suptext{S}}_t  \EM_t  \left  [  \frac{ A_{T_1}^{\suptext{S}} }{A_{ t}^{\suptext{S}} } Z_{T_1} \right  ]   \prod_{k=2}^N   \EM  \left  [  \frac{ A_{T_k}^{\suptext{S}} }{A_{T_{k-1}}^{\suptext{S}} } Z_{T_k} \right  ]    \\    
					     &\hspace{4.5cm}+ C^{\suptext{LPI}}_ {T_0} \frac{c_5 }{h_{t}^{\suptext{N}}}   A^{\suptext{R}}_{t} A^{\suptext{S}}_t \EM_t \left  [  \frac{  A_{T_1}^{\suptext{R}} A_{T_1}^{\suptext{S}} }{ A_{ t}^{\suptext{R}} A_{ t}^{\suptext{S}} } Z_{T_1} \right  ] \prod_{k=2}^N   \E^{\M}  \left  [  \frac{ A_{T_k}^{\suptext{R}}  A_{T_k}^{\suptext{S}}  }{ A_{   T_{k-1}}^{\suptext{R}} A_{  T_{k-1}}^{\suptext{S}} } Z_{T_k} \right  ]   . 
\end{split}
\end{align}
All the expectations can be calculated by Lemma~\ref{RILemKey2}, in the appendix, to obtain
\begin{align*}
\EM_t \left  [     \frac{  A_{T_1}^{\suptext{R}} A_{T_1}^{\suptext{S}} }{ A_{ t}^{\suptext{R}} A_{ t}^{\suptext{S}} } Z_{T_1}  \right   ]
 &=\EM  \left . \left   [  \frac{  A_{T_1}^{\suptext{R}} A_{T_1}^{\suptext{S}} }{ A_{ t}^{\suptext{R}} A_{ t}^{\suptext{S}} }   \left ( (1+K_c )  -\left (1+ K_c -x \frac{C_{T_1}}{C_t} \right )^{+}+ \left ( 1+ K_f -x \frac{C_{T_1}}{C_t}\right )^+ \right )  \right   ]    \right   \lvert_{x= C_t/C_{T_0}}  \\
&=\beta_c  \exp\left ( \kappa_{t  T_1} (  w_{\subtext{S}} +w_{\subtext{R}}) \right ) +   \frac{1}{\pi} \int_{\R ^+ }  \Re \frac{ \beta_c \, \vartheta^{21}(\alpha^{1c}_t,u) +\beta_f \, \vartheta^{21}(\alpha^{1f}_t,u) }{(R+iu) (1+R+iu) }  \diff u.
\end{align*}
The remaining remaining  expectations are calculated in the same way. 
\end{proof} 
We note that each $V$ is calculated like a ZC floor or YoY caplet, i.e. the evaluation is no more complicated than for a YoY cap. To price  multiple LPI-linked ZC bonds, the shorter maturity bond prices can be found from the factors needed for the longer maturity ones.
 
\subsection{Gaussian formulae} \label{RISecGaussForm}
In this section we  derive the results equivalent to  Theorems~\ref{RITheoYoYFourier}, \ref{RITheoZCFourier}  and \ref{RITheoLPIBondFourier} under the assumption of the model in Specification~\ref{RIExTimeChangeWP}. The results will be Black-Scholes-style formulae.

\begin{Proposition}  \label{RIPropYoYLN}   
Assume the additive exponential-rational pricing kernel where $(X_t)_{0 \leq t}$ is the time-changed Wiener process of Specification~\ref{RIExTimeChangeWP}. Assume that $t \leq T_{i-1} < T_i$ and denote by  $V^{\suptext{YoYFl}}_{t} $ the price of the floorlet, as in Theorem~\ref{RITheoYoYFourier}. Then 
\begin{equation*}
\begin{aligned}
V^{\suptext{YoYFl}}_{t} &=\frac{c_0 K }{h^{\suptext{N}}_t   }   A_t^{\suptext{S}}  \left (  \e^{\delta_1} \normcdf  ( - d_1 )  -   \alpha  \e^{\delta_1 + \mu_y + \tfrac{1}{2} ( 1+ 2b_1 )  \sigma_y^2 } \normcdf ( -( d_1+ \sigma_y ) )   \right )      \\ 
					      &+\frac{c_6 K }{h_t^{\suptext{N}}}  A_t^{\suptext{R}}  A_t^{\suptext{S}} \left (  \e^{\delta_2  } \normcdf  (  -d_2)   -   \alpha    \e^{\delta_2 + \mu_y + \tfrac{1}{2} ( 1+ 2b_2 )  \sigma_y^2  } \normcdf  ( - ( d_2 +\sigma_y ))   \right ) ,  
\end{aligned}
\end{equation*}
where 
\begin{align*} 
  \begin{aligned}
&   \mu_{x_1} ={}   \innerp{w_{\subtext{S}}}{\mu_{T_i} - \mu_{T_{i-1}}} ,  \\   
   &   \sigma_{x_1}   = {}   \innerp{w_{\subtext{S}}}{ ( \Sigma_{T_i} -\Sigma_{T_{i-1}}) w_{\subtext{S}}} ,  \\   
    &   \mu_{y}  = {}   - \innerp{w_{\subtext{S}}}{\mu_{T_i} - \mu_{T_{i-1}}} ,\\
    &  \sigma_{x_1 y}   = {}   -\innerp{w_{\subtext{S}}}{  ( \Sigma_{T_i} -\Sigma_{T_{i-1}} ) w_{\subtext{S}}} ,
  \end{aligned}
  &&
  \begin{aligned}
&  \mu_{x_2}  = {}   \innerp{w_{\subtext{R}}+w_{\subtext{S}}}{\mu_{T_i} - \mu_{T_{i-1}}} , \\
 &         \sigma_{x_2}   = {}   \innerp{w_{\subtext{R}}+w_{\subtext{S}}}{ ( \Sigma_{T_i} -\Sigma_{T_{i-1}}) (w_{\subtext{R}}+w_{\subtext{S}  } ) },  \\
 &       \sigma_{ y}   = {}   \innerp{w_{\subtext{S}}}{ ( \Sigma_{T_i} -\Sigma_{T_{i-1}} ) w_{\subtext{S}}} ,   \\
 &         \sigma_{x_2 y}   = {}  - \innerp{w_{\subtext{S}}}{ ( \Sigma_{T_i} -\Sigma_{T_{i-1}}) ( w_{\subtext{R}}+w_{\subtext{S}  } ) },
  \end{aligned}
 \end{align*}
 see Specification~\eqref{RIExTimeChangeWP} for the values of $\mu_t$ and $\Sigma_t$, 
 and, for $j=1,2$,
$\delta_j =  a_j + b_j \mu_y + (b_j \sigma_y)^2/2$, $d_j =\tfrac{1}{\sigma_y}  (  \ln \alpha + b_j\sigma_y^2 + \mu_y)$,
$a_j =\mu_{x_j} - \frac{\sigma_{x_j y}}{\sigma_y^2}\mu_y - (\sigma_{x_j}^2-\tfrac{\sigma^2_{x_j y}}{\sigma_y^2})/2$, $b_j=\tfrac{\sigma_{x_j y}}{\sigma_y^2}$.
Moreover,
 \begin{alignat*}{2}
\alpha = {} & \frac{ S(T_{i-1})   } { K  S(T_i)    } ,   \quad c_0= {}  R(T)  \left(1-b^{\suptext{R}}(T)\right)   S(T),   \\
c_6 ={} & R(T) b^{\suptext{R}}(T) S(T)  \exp \left ( \kappa_{ t T_{i-1} } (   w_{\subtext{R}}+w_{\subtext{S}}  ) \right ) \exp \left ( \kappa_{ T_{i} T } (   w_{\subtext{R}}+w_{\subtext{S}}  ) \right ).
  \end{alignat*}   
\end{Proposition}
\begin{proof}
Using the independent increments property 
 \begin{align*}
V^{\suptext{YoYFl}}_{t}  =  &  {}   \frac{1}{h^{\suptext{N}}_t} \EM_t \left [ R(T)  \left( 1+ b^{\suptext{R}}(T) ( A_T^{\suptext{R}} -1)\right) S(T) A_T^{\suptext{S}} \left  ( K- \frac{ C_{T_i} }{C_{T_i-1}}  \right  )^+ \right  ] \notag{} \\
			   = &  {}     \frac{c_0 K }{h^{\suptext{N}}_t   }  A_t^{\suptext{S}}  \EM \left  [    \frac{ A_{T_i}^{\suptext{S}}}{A_{T_{i-1}}^{\suptext{S}}} \left (  1- \frac{1}{K } \frac{ C_{T_i} }{C_{T_i-1}}  \right )^+ \right ] +  \frac{c_6 K }{h_t^{\suptext{N}}}  A_t^{\suptext{R}}    A_t^{\suptext{S}} \EM  \left  [  \frac{  A_{T_i}^{\suptext{R}} A_{T_i}^{\suptext{S}} }{A_{T_{i-1}}^{\suptext{R}} A_{T_{i-1}}^{\suptext{S}} }   \left (1-    \frac{1}{K } \frac{ C_{T_i} }{C_{T_i-1}} \right  )^+ \right ]  
\end{align*}
Now we apply  Lemma~\ref{RILemBS2}  to each term to obtain the result.
\end{proof}
  The case where $T_{i-1} < t < T_{i}$ is derived similarly, see Proposition~\ref{RIPropLPILN}. 
The price formula for the ZC floor is derived analogously.
\begin{Proposition}  \label{RIPropZCLN} 
Assume the additive exponential-rational pricing kernel where $(X_t)_{0 \leq t}$ is the time-changed Wiener process of Specification~\ref{RIExTimeChangeWP}.
Let  $T_0 \leq t < T_i \leq T$, then the price of the ZC floor is
\begin{equation*}
\begin{aligned}
V^{\suptext{ZCFl}}_{t}  =  &  {}   \frac{ c_0}{h^{\suptext{N}}_t}\frac{1}{S(T_i)}  \left(    \alpha_t \,\e^{ \mu_y + \tfrac{1}{2}  \sigma_y^2  } \normcdf  (d^1_t + \sigma_y)  -\normcdf   (  d^1_t )    \right)  +   \frac{ c_7}{h^{\suptext{N}}_t}\frac{1}{S(T_i)}   A_{t}^{\suptext{R}}    \left (     \alpha_t \,\e^{\delta + \mu_y + \tfrac{1}{2} (1+2 b) \sigma_y^2  } \normcdf  (   d^2_t +\sigma_y  ) - \e^{\delta  } \normcdf  (  d^2_t)   \right )   
\end{aligned}
\end{equation*}
where
$\mu_{x}  =   \innerp{w_{\subtext{R}}}{\mu_{T_i} - \mu_t}$, $\sigma_{x}   =   \innerp{w_{\subtext{R}}}{ ( \Sigma_{T_i}  -\Sigma_t)  w_{\subtext{R}}}$,  $\mu_{y}  =    \innerp{w_{\subtext{S}}}{\mu_{T_i} - \mu_t}$, $\sigma_{ y}   =   \innerp{w_{\subtext{S}}}{ ( \Sigma_{T_i}  -\Sigma_t ) w_{\subtext{S}}}$, $\sigma_{x y}   =  \innerp{w_{\subtext{R}}}{ ( \Sigma_{T_i} -\Sigma_t ) w_{\subtext{S}}}$,
 see Specification~\eqref{RIExTimeChangeWP} for the values of $\mu_t$ and $\Sigma_t$. Furthermore,
  $d^1_t = \tfrac{1}{\sigma_y}  (  \ln \alpha_t +\mu_y)$,  $d^2_t =\tfrac{1}{\sigma_y}  (  \ln \alpha_t + b \sigma_y^2 + \mu_y)$, $a=\mu_{x} - \frac{\sigma_{x y}}{\sigma_y^2}     \mu_y - \frac{1}{2} (\sigma_{x}^2-\tfrac{\sigma^2_{xy}}{\sigma_y^2})$,  $b=\tfrac{\sigma_{xy}}{\sigma_y^2}$,  $\delta = a+ b \mu_y + \frac{(b \sigma_y)^2}{2}$,     
$\alpha_t =K S(T_i) A_t^{\suptext{S}}$,  $c_0=   R(T) \left(1-b^{\suptext{R}}(T)\right) S(T)$,  $c_7 = R(T) b^{\suptext{R}}(T)   S(T)  \e^{\kappa_{T_i T} ( w_{\subtext{R}}+w_{\subtext{S}} ) }$.    
\end{Proposition}
\begin{proof}
Using the properties of the conditional expectation and the independence of the increments, we may write 
 \begin{align*}
V^{\suptext{ZCF}}_{t}   =  &  {} \frac{1}{h^{\suptext{N}}_t} \EM_t \left [ R(T)  \left( 1+ b^{\suptext{R}}(T) ( A_T^{\suptext{R}}-1) \right) S(T) A_T^{\suptext{S}} \left ( K- C_{T_i}    \right )^+  \right ] \notag{} \\
			  = &  {}  \left.   \frac{ c_0}{h^{\suptext{N}}_t}\frac{1}{S(T_i)}    \EM \left [ \left (  \frac{1}{x}   K   \frac{ C_t }{C_{T_i}} -1 \right  )^+\right  ] \right    \lvert_{x=C_t}  +\left. \frac{ c_7}{h^{\suptext{N}}_t}\frac{ A_{t}^{\suptext{R}} }{S(T_i)}    \EM \left [  \frac{A_{T_i}^{\suptext{R}}}{A_t^{\suptext{R}} } \left (   \frac{1}{x}   K   \frac{ C_t }{C_{T_i}} -1 \right  )^+ \right ]    \right \lvert_{x= C_t} .    
\end{align*}
Now applying  Lemma~\ref{RILemBS}  to each term yields the result.
\end{proof}
The formula for the price process of the LPI-linked ZC bond follows in the same way.
\begin{Proposition}[Limited price index bond] \label{RIPropLPILN}
Assume the additive exponential-rational pricing kernel with $(X_t)_{0 \leq t}$ the time-changed Wiener process of Specification~\ref{RIExTimeChangeWP}.   Assume without loss of generality that $T_0 \leq t < T_1< T_2 <\dots < T_N \leq T$. 
Then
\begin{equation*}
P_{tT}^{\suptext{LPI}}  = \tfrac{ 1}{h_{t}^{\suptext{N}}}  C^{\suptext{LPI}}_ {{T}_0} \left (   c_0 V^{11}_t  \prod_{k=2}^N V^{1k}  + c_5 A_t^{\suptext{R}} V_t^{21}  \prod_{k=2}^N V^{2k}  \right )  A_t^{\suptext{S}} 
\end{equation*}
where 
$c_0=R(T) (1-b^{\suptext{R}}(T) ) S(T)$,  $c_5 =R(T) b^{\suptext{R}}(T) S(T)\exp\left( \kappa_{ T_N T } (   w_{\subtext{R}}+w_{\subtext{S}}  ) \right),$ and for $k=2,\dots,N$, $j=1,2$
\begin{equation*}
\begin{aligned}
V^{j1}_t= {} & \left (1+K_c \right )  \e^{\kappa_{t  T_1} ( \omega^j) } +     \e^{\delta^{j1}  } \normcdf  \left ( - d_t^{j1c}  \right  )    - \alpha_t^{1c}    \e^{\delta^{j1} + \mu^1_y + \tfrac{1}{2} \left  (1+2b^{j1}  \right )  \left ( \sigma^1_y \right )^2 } \normcdf \left  (   - d_t^{j1c}  -\sigma^1_y  \right  )   \\
{} & +     \e^{\delta^{j1}  } \normcdf  \left ( -  d_t^{j1f}  \right  )    - \alpha_t ^{1f}    \e^{\delta^{j1} + \mu^1_y + \tfrac{1}{2}  \left (1+2b^{j1}  \right )   \left ( \sigma^1_y \right )^2 } \normcdf  \left (  -  d_t^{j1f}  - \sigma^1_y    \right )    \\
V^{jk}= {} & \left (1+K_c \right )  \e^{\kappa_{ T_{k-1}  T_k}  ( \omega ^j )} +     \e^{\delta^{jk}  } \normcdf  \left  ( -  d^{jkc}  \right  )    - \alpha^{kc}    \e^{\delta^{jk} + \mu^k_y + \tfrac{1}{2} \left  (1+2b^{jk} \right )  ( \sigma^k_y)^2  } \normcdf  \left (  -    d^{jkc}  - \sigma^k_y  \right  )   \\
{} & +     \e^{\delta^{jk}  } \normcdf  \left  ( -   d^{jkf}  \right )   - \alpha^{kf}    \e^{\delta^{jk} + \mu^k_y +\tfrac{1}{2} \left (1+2b^{jk} \right  )  ( \sigma^k_y)^2  } \normcdf  \left  (  -  d^{jkf}  -\sigma^k_y  \right), 
\end{aligned}
\end{equation*}
where, for $k=1,\dots,N$,
\begin{align*} 
  \begin{aligned}
  & \mu^k_{x_1} =   \innerp{w_{\subtext{S}}}{\mu_{T_k} - \mu_{T_{k-1} \vee  t}},   \\   
   &   \sigma^k_{x_1}   =  \innerp{w_{\subtext{S}}}{ ( \Sigma_{T_k} -\Sigma_{T_{k-1}\vee  t}) w_{\subtext{S}}},   \\   
    &   \mu^k_{y}  =  - \innerp{w_{\subtext{S}}}{\mu_{T_k} - \mu_{T_{k-1} \vee  t}} ,\\
    &  \sigma^k_{x_1 y}   =   -\innerp{w_{\subtext{S}}}{ ( \Sigma_{T_k} -\Sigma_{T_{k-1}\vee  t} ) w_{\subtext{S}}},  \\
    &  \omega^1   =  w^{\suptext{S}}, 
  \end{aligned}
  &&
  \begin{aligned}
  & \mu^k_{x_2}  =   \innerp{w_{\subtext{R}}+w_{\subtext{S}}}{\mu_{T_k} - \mu_{T_{k-1} \vee  t }}  ,\\
    &      \sigma^k_{x_2}   =   \innerp{w_{\subtext{R}}+w_{\subtext{S}}}{  ( \Sigma_{T_k} -\Sigma_{T_{k-1} \vee  t } ) (  w_{\subtext{R}}+w_{\subtext{S}} ) } , \\
     &   \sigma^k_{ y}   =  \innerp{w_{\subtext{S}}}{  ( \Sigma_{T_k} -\Sigma_{T_{k-1}\vee  t} ) w_{\subtext{S}}}  ,   \\
      &    \sigma^k_{x_2 y}   =  - \innerp{w_{\subtext{S}}}{ (\Sigma_{T_k} -\Sigma_{T_{k-1}\vee  t} )  ( w_{\subtext{R}}+w_{\subtext{S}} ) } , \\
          &  \omega^2   = w^{\suptext{S}}+ w^{\suptext{R}}, 
  \end{aligned}
 \end{align*}
   see Specification~\eqref{RIExTimeChangeWP} for the values of $\mu_t$ and $\Sigma_t$. For $j=1,2$, $k=1,\dots,N$, we have
\begin{align*}
  &\delta^{jk}  = a^{jk} + b^{jk}  \mu^k_y + \frac{(b^{jk}  \sigma^k_y)^2}{2}, &b^{jk} =   \frac{\sigma^k_{x_jy}}{(\sigma_y^k)^2},& &a^{jk}  = \mu^k_{x_j} - \frac{\sigma^k_{x_j y}}{(\sigma^k _y)^2}\mu^k_y - \frac{1}{2} \left ((\sigma^k_{x_j} ) ^2-\tfrac{( \sigma^k_{x_j y})^2}{(\sigma^k_y)^2} \right ),&
\end{align*} 
 for $k=2,\dots,N$ and  $l = c,f$
\begin{align*}
 d_t^{j1l}  = {} &  \tfrac{1}{\sigma^k_y}  \left ( \ln \alpha^{1l}_t  +   b^{1j}  (\sigma^1_y)^2 + \mu^1_y \right ) ,  &   
\alpha^{1l}_t = {} & \frac{  S(T_1) }{ (1+K_l)  S( T_0 )} \frac{ A^{\suptext{S}}_{T_0 } }{ A_t^{\suptext{S}}    } ,  \\
d^{jkl}  = {} &  \tfrac{1}{\sigma^k_y}  \left (\ln \alpha^{kl} + b^{jk}   (\sigma^k_y)^2 + \mu^k_y \right )   , 
 &  \alpha^{kl}  = {} & \frac{ S(T_{i-1})   } {  (1+K_l)   S(T_i)    } .
\end{align*}
\end{Proposition}
\begin{proof}
We may use the proof of Theorem~\ref{RITheoLPIBondFourier} up to  the expectations in Eq.~\eqref{RITheoLPIEqExp} which can be calculated by Lemma~\ref{RILemBS2} in the Appendix.  
\end{proof}

\subsection{Nominal products}   \label{RISecNomProd}
An important nominal linear interest rate derivative is the swap which pays the difference between a fixed rate and a floating rate. Loosely speaking we refer to this rate as the  LIBOR. Suppose we have a sequence of  time points $T_0<T_1<\dots< T_N$, and let $\delta_i=T_i-T_{i-1}$.  A payer's swap pays  $\delta_i(L(T_i,T_{i-1},T_i)-K)$ at each $T_i$, where $L(T_i,T_{i-1},T_i)$ is the \textit{LIBOR} spot rate. We assume for ease of exposition that payments on the fixed leg $K$ and floating leg $L(T_i,T_{i-1},T_i)$ both occur at time $T_i$. It follows that the price of the swap at time $t\leq T_0$ is given by 
\begin{equation} \label{RIEqswapmodelfree}
V^{\suptext{Sw}}_{t} =\sum_{i=1}^N \delta_i\left(  \EM_t\left[ \frac{h_{T_i}^{\suptext{N}}}{h_{t}^{\suptext{N}}}L(T_i,T_{i-1},T_i)\right ] -KP^{\suptext{N}}_{tT_i}\right ) .  
\end{equation}
\begin{Definition}[Single-curve setup]
If LIBOR rates are spanned by a single system of nominal bonds for all tenors, we say that we are in the single-curve setup.
\end{Definition}  
We refer to \citep{GrbacRunggaldier2016} for an overview of single- and multi-curve interest rate models.  Within the single-curve setup, we have the no-arbitrage relation 
\begin{equation*}
L(T_i,T_{i-1},T_i)= \frac{1}{\delta_i}\left(\frac{1}{P^{\suptext{N}}_{T_{i-1} T_i}}-1\right).
\end{equation*}
It follows that
\begin{equation*}
\EM_t\left[ \frac{h_{T_i}^{\suptext{N}}}{h_{t}^{\suptext{N}}}L(T_i,T_{i-1},T_i)\right]=\frac{1}{\delta_i}\left(P_{t T_{i-1}}^{\suptext{N}}-P^{\suptext{N}}_{tT_i}\right),
\end{equation*}
and thus the swap price is given  by
\begin{equation} \label{RIeqSwap1Curve}
V^{\suptext{Sw}}_t =\sum_{i=1}^N \left[ P^{\suptext{N}}_{t T_{i-1}} -(1+\delta_i K)P^{\suptext{N}}_{tT_i}\right]. 
\end{equation}

\subsubsection{Swaptions}  \label{RISecSwaptions} 
A swaption is an option to enter  a swap at  some future time. If we let this point in time be  $T_k$ and denote the maturity of the underlying swap by $T_N$, the swaption price at $t\leq T_k$ is given by Eq.~\eqref{RIeqPricingFormulaM} and we have
\begin{equation} \label{RIEqSwaptionPrice}
V^{\suptext{Swn}}_{t } = \frac{N }{h_t^{\suptext{N}}}\EM_t \left [ h^{\suptext{N}} _{T_k}   \left( V^{\suptext{Sw}} _{T_k}\right) ^+ \right],  
\end{equation}
where $N$ is the notional.

\begin{Proposition} \label{RIPropSingleCurveSwaption}
Assume an additive exponential-rational pricing kernel model and assume the single curve setup.  Let  $V^{\suptext{Swn}}_{t}  $,  as in  \eqref{RIEqSwaptionPrice}, be the swaption price at time $0\leq t \leq T_k$. 
Let $T_{k+1}<T_{k+2}<\dots <T_N$ denote the payment dates of the underlying swap. Set 
\begin{align*}
c_0 ={} &  \sum_{i=k+1}^N  \left(   R(T_{i-1}) \left(1-b^{\suptext{R}}(T_{i-1}) \right) S(T_{i-1}) - (1+ \delta_i K )R(T_{i})\left(1-b^{\suptext{R}}(T_{i})\right)     S(T_{i })    \right),  \\
c_1 = {} &  \sum_{i=k+1}^N      \Bigl (    R(T_{i-1})  b^{\suptext{R}}(T_{i-1} )   S(T_{i-1}) \e ^{ \kappa_{T_kT_{i-1} } (  w_{\subtext{R}}+w_{\subtext{S}}) } - (1+ \delta_i K ) R(T_{i})  b^{\suptext{R}}(T_{i} )  S(T_{i })    \e ^{ \kappa_{T_kT_{i} } (  w_{\subtext{R}}+w_{\subtext{S}})  }    \Bigr ). 
\end{align*}
If $c_0<0$ and $c_1<0$, then   $ V^{\suptext{Swn}}_{t}=0$, and if $c_0>0$ and $c_1>0$, then  
\begin{equation*}
V^{\suptext{Swn}}_{t} =\frac{1}{h_t^{\suptext{N}}}  A_t^{\suptext{S}} \left ( c_0 +A_t^{\suptext{R}}  c_1 \e^{\kappa_{tT_k} (  w_{\subtext{R}} + w_{\subtext{S}}  ) } \right )  .
\end{equation*} 
If $\sign(c_0 ) \neq \sign (c_1)$, define
$Y_1=    \innerp{  w_{\subtext{S}}}{ X_{T_k} }$, $Y_2=    \innerp{  w_{\subtext{R}} }{ X_{T_k} }$ and $q_t(z)= \EM_t[\e^{  \innerp{ z}{  (Y_1,Y_2) } }] $.
Let  $R<-1$ if $c_0<0 $  and  $R>0$ if $c_0>0$. Assume that  $q_t(-R,1) < \infty$.
Then
\begin{equation*}
 V^{\suptext{Swn}}_{t} = \frac{\abs{c_0}  }{\pi h_t^{\suptext{N}}}  \int_{\R^+}    \Re \frac{  \vartheta_t(u) }{ (R+iu)(1+R+iu)}   \diff u  
\end{equation*}
where
$
 \vartheta_t(u)  = \alpha^{-(R+iu)}   q_t(-(R+iu) ,1 ) 
$
and $\alpha = \abs{c_1/c_0}$.
\end{Proposition}
\begin{proof}
From the nominal bond pricing formula~\eqref{RIeqNomBondPrice} we have that
\begin{equation*}
h_{T_k}^{\suptext{N}} P_{T_k T_i}^{\suptext{N}} = R(T_i ) S(T_i)\left ( A_{T_k}^{\suptext{S}} \left(1-b^{\suptext{R}}(T_i)\right)+  b^{\suptext{R}}(T_i) \EM_{T_k} \left[A^{\suptext{R}}_{T_i} A_{T_i}^{\suptext{S}}\right] \right) .
\end{equation*}
Then, inserting this into the swap formula~\eqref{RIeqSwap1Curve}, we obtain 
\begin{align*}
h_{T_k}^{\suptext{N}} V^{\suptext{Sw}}_{T_k}  = {} & \sum_{i=k+1}^N  \bigg (R(T_{i-1}) S(T_{i-1})\left( A_{T_k}^{\suptext{S}}  \left(1-b^{\suptext{R}}(T_{i-1})\right)  + b^{\suptext{R}}(T_{i-1}) \EM_{T_k} \left[A^{\suptext{R}}_{T_{i-1}} A_{T_{i-1}}^{\suptext{S}}\right]\right)\\
{} &  - (1+\delta_i K )   R(T_{i}) S(T_{i})  \left( A_{T_k}^{\suptext{S}} \left(1-b^{\suptext{R}}(T_{i})\right)  + b^{\suptext{R}}(T_{i}) \EM_{T_k} \left[A^{\suptext{R}}_{T_i} A_{T_i}^{\suptext{S}} \right]\right)\bigg) .
\end{align*}
We may write 
\begin{equation} \label{RIEqExpecProdAdd}
\EM_{T_k} \left[A^{\suptext{R}}_{T_i} A_{T_i}^{\suptext{S}}\right] =A^{\suptext{R}}_{T_k}   A^{\suptext{S}}_{T_k}  \e^{ \kappa_{T_kT_i } (  w_{\subtext{R}} +w_{\subtext{S}}) }. 
\end{equation} 
Collecting terms and using the fact that $A^{\suptext{S}}_{t}>0$ for any $t \geq 0$, we arrive at
\begin{equation*}
\left(h^{\suptext{N}}_{T_k} V^{\suptext{Sw}}_{T_k}\right)^+  = A_{T_k}^{\suptext{S}} \left( c_0    + c_1  A_{T_k}^{\suptext{R}}   \right)^+ .
\end{equation*}
 If $c_0>0$ and $c_1<0$, then 
\begin{equation} \label{RIEqSCSwaption1}
  \left( c_0    + c_1  A_{T_k}^{\suptext{R}} \right)^+  =  \abs{  c_0 }    \left( 1    - \alpha  A_{T_k}^{\suptext{R}}   \right)^+, 
\end{equation}
where we recall that  $\alpha = \abs{c_1/c_0}$. The result follows from Lemma~\ref{RILemKey3}.
If $c_0<0$ and $c_1>0$, then
\begin{equation} \label{RIEqSCSwaption2}
  \left( c_0    + c_1  A_{T_k}^{\suptext{R}}   \right)^+  = \abs{  c_0 }  \left( \alpha A_{T_k}^{\suptext{R}}  -1   \right)^+ 
\end{equation}
and the result follows from  Lemma~\ref{RILemKey2}. The two remaining cases are straightforward.
\end{proof}
The independent increments property of $(X_t)$ is only used to obtain Eq.~\eqref{RIEqExpecProdAdd}.

\begin{Remark}
Under the assumption of  Specification~\ref{RIExTimeChangeWP} the counterpart to Proposition~\ref{RIPropSingleCurveSwaption} is obtained by applying  Lemma~\ref{RILemBS} and \ref{RILemBS2} to  Eq.~\eqref{RIEqSCSwaption1} and Eq.~\eqref{RIEqSCSwaption2}. 
\end{Remark}

\subsubsection{Multi-curve interest rate setting} \label{RISecMultiCurve}
We can, at a relatively low cost, allow our model to incorporate multi-curve-features. This is done by modelling~\eqref{RIEqswapmodelfree} as a rational function of state variables not fully spanned by the ones driving the nominal bonds. We model  the  forward LIBOR by
\begin{equation*}
L(t,T_{i-1},T_{i}):= \frac{1}{h_t^{\suptext{N}}} \EM_t\left[ h_{T_i}^{\suptext{N}}  L(T_i,T_{i-1},T_i)\right].
\end{equation*}
\citep{CrepeyEtAl2016} propose the following definition, which we shall adopt.
\begin{Definition}[Rational multi-curve setup]
Let 
\begin{equation} \label{RIeqMultiCurveL}
L(t,T_{i-1},T_{i}):= \frac{L(0,T_{i-1},T_i)+b^{\suptext{L}}(T_{i-1},T_i) ( A_t^{\suptext{L}}-1 ) }{h_t^{\suptext{N}}}, 
\end{equation}
where $(A_t^{\suptext{L}})_{0\leq t}$ is an $\M$-martingale with $A_0^{\suptext{L}}=1$, and where $b^{\suptext{L}}(\cdot,\cdot)$ and  $L(0,\cdot,\cdot)$ are deterministic functions.
\end{Definition}

We consider
$
A^{\suptext{L}} _t = \exp( \innerp{w_{\subtext{L}}}{X_t}),
$
where $w_{\subtext{L}}$ is chosen such that $(A^L_t)$ is a martingale. This is analogous to how $(A_t^{\suptext{R}})$ and $(A_t^{\suptext{S}})$ are modelled. Adding a multi-curve dimension to the nominal markets has no effect  on any of the formulae derived for the inflation products.  It does though impact the swaption formula.
\begin{Proposition} \label{RIPropMultiCurveSwaption} 
Assume an additive exponential-rational pricing kernel model and the multi-curve setup.  Consider a swaption with maturity $T_k$ written on a swap with payments dates   $T_{k+1}<T_{k+2}<\dots <T_N$.  The swaption price $V^{\suptext{Swn}}_{t}  $ at  $0 \leq t\leq T_k$  is given by
\begin{equation} \label{RIeqSwaptionMean}
V^{\suptext{Swn}} _{t }  = \frac{1}{h_t^{\suptext{N}}}  \int_{\R^n} H^{\suptext{M}}(x) ^+ \intspace  m_{T_k} (\dd x)   
\end{equation}
where $m_{T_k}$ is the distribution of $X_{T_k}$ and
\begin{equation*}
H^{\suptext{M}}(x) =      c_0 + c_{\subtext{L}}  \exp\left ( \innerp{w_{\subtext{L}}}{ x} \right ) +  c_{\subtext{S}} \exp \left ( \innerp{w_{\subtext{S}}}{x} \right  ) + c_{SR}  \exp \left ( \innerp{w_{\subtext{R}} +w_{\subtext{S}}}{ x} \right ).
\end{equation*}
Furthermore,
\begin{align*}
&c_0=   \sum_{i=k+1}^N   \delta_i   \left (L(0,T_{i-1},T_i)-b^{\suptext{L}}(T_{i-1},T_i) \right ),&  &c_{\subtext{L}} = \sum_{i=k+1}^N    \delta_i   b^{\suptext{L}}(T_{i-1},T_i)&  \\
&c_{\subtext{S}} =  - \sum_{i=k+1}^N    \delta_i K       R(T_{i}) S(T_{i}) \left (1-b^{\suptext{R}}(T_{i}) \right),&  
&c_{\subtext{SR}} = -\sum_{i=k+1}^N    \delta_i K    R(T_{i})  b^{\suptext{R}}(T_{i})   S(T_{i})  \e^{\kappa_{T_k T_i} (  w_{\subtext{R}} +w_{\subtext{S}})  }.&
\end{align*}
\end{Proposition}
\begin{proof}
 Using~\eqref{RIeqMultiCurveL} we may write
\begin{align*}
h_{T_k}^{\suptext{N}} V^{\suptext{Sw}} _{T_k}  = {} & \sum_{i=k+1}^N  \delta_i\bigg(L(0,T_{i-1},T_i) + b^{\suptext{L}}(T_{i-1},T_i) (A^{\suptext{L}}_{T_k} -1) \\ 
&\hspace{3cm} -  K  R(T_{i})  S(T_{i}) \left( A_{T_k}^{\suptext{S}} \left(1-b^{\suptext{R}}(T_{i})\right)  + b^{\suptext{R}}(T_{i}) \EM_{T_k} \left[A^{\suptext{R}}_{T_i} A_{T_i}^{\suptext{S}} \right]\right)\bigg). 
\end{align*}
By collecting the terms, the result follows.
\end{proof}
We note that, analogous to the single-curve setup, the independent increments of $(X_t)$ are only used to evaluate $\EM_{T_k} [A^{\suptext{R}}_{T_i} A_{T_i}^{\suptext{S}}]$. In our applications, since  $(X_t)$ is bi-variate, we may apply  the two-dimensional cosine method of \citep{RuijterOosterlee2012}. 
The most immediate method for handling~\eqref{RIeqSwaptionMean} in higher dimensions is in the style of \citep{SingletonUmantsev2002}, where 
$
H(x)^+ \approx  H(x)  \1{G}  .
$
with  $G = \{  H(x) \geq  0 \}$ being exact.  If $G = \{ \innerp{\omega}{X_{T_k}} > \alpha\}$, this leads to a one-dimensional integral, see \citep{Kim2014} and  \citep{CuchieroEtAl2019}. If $G = \{ \beta_1  \exp( \innerp{ \omega_1}{x} ) + \beta_2 \exp ( \innerp{ \omega_2}{x})  > \alpha \}$  the inversion formula becomes a two-dimensional \citep{HurdZhou2010}-type formula.
\section{Calibration examples} \label{RISecCaliExam}
In this section, we show the calibration properties of the models on real data.  We consider EUR data from Bloomberg from 1 January 2015. The necessary data consists of OIS zero-yields constructed from EONIA overnight indexed swaps, LIBOR discrete curves based on EURIBOR and a term structure of ZC forward rates, as well as YoY cap and floor prices and EURIBOR swaptions. There is no LPI traded on EUR data. The OIS and EURIBOR curves are constructed directly in the Bloomberg system. We then set the nominal curve equal to the OIS curve, and the initial real (or equivalently the initial inflation-linked) curve is implied from the OIS curve and zero-coupon inflation forward rate using the methodology described in Section~\ref{RISecFundProd}. The prices for YoY caps and floors are available to us for maturities 2, 5, 7, 10, 12, 15, 20 and 30 years. The strikes for the floors range from -1\% to 3\% and caplets from 1\% to 6\%. Quotes for YoY swap rates are not available to us, but the overlap in strikes for YoY caps and floors allows us to use put-call (cap-floor) parity to imply YoY swap rates consistent with the option prices.

\begin{figure}[htbp!]
\centering
\includegraphics[scale=0.49]{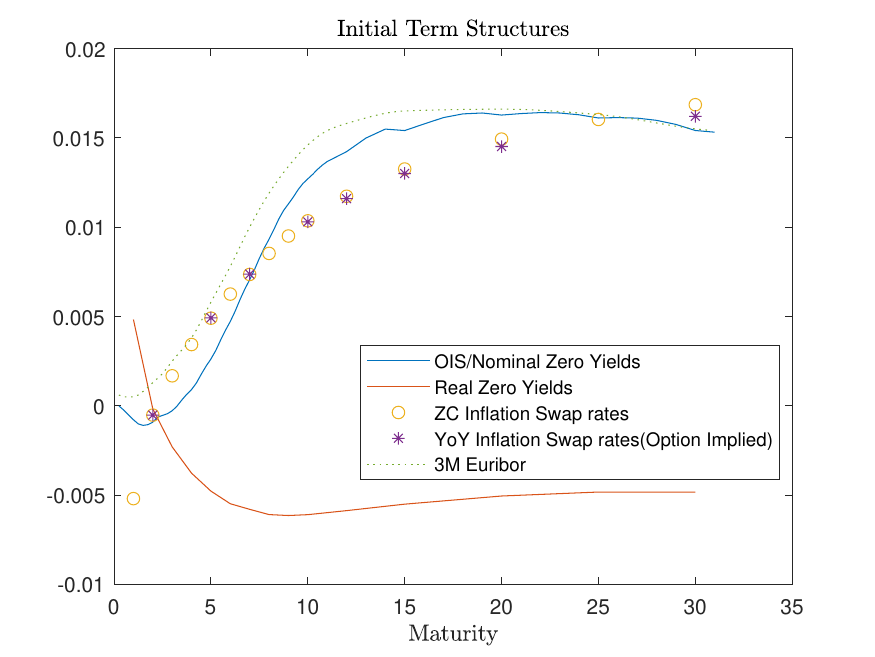}\includegraphics[scale=0.5]{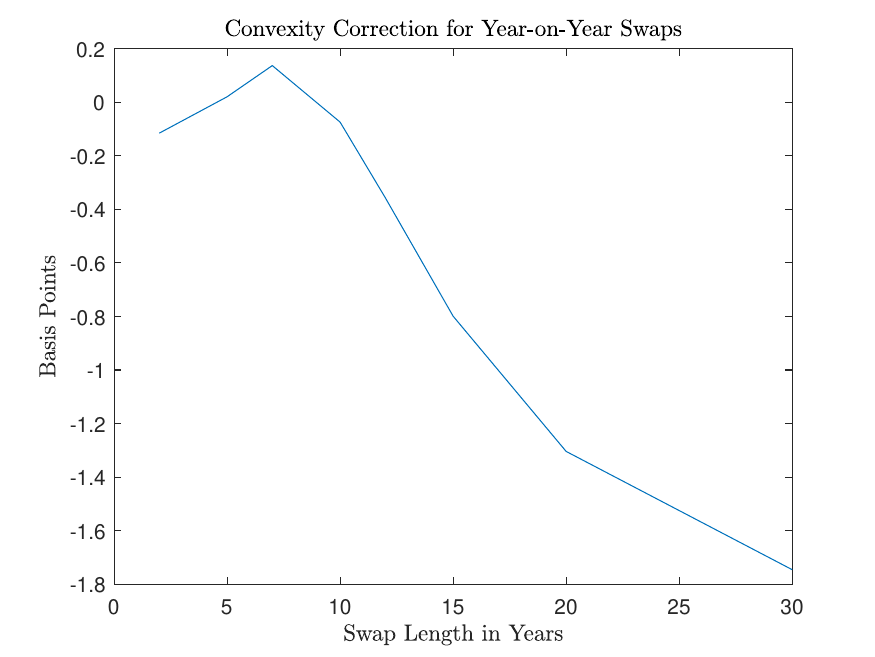}
\caption{Left: The initial curves. Right: Convexity corrections for YoY swap rates. 1 January 2015.  }
\label{RIf:Discrete rates}
\end{figure}
In Figure~\ref{RIf:Discrete rates}, all the curves are plotted on the left-hand side. The real curve is plotted as zero-coupon rates, and we note that on this day, there is a consistently negative real curve with a widening gap to the nominal as the maturity increases. The 3m EURIBOR and OIS curve are plotted as discrete forward rates with 3m increments to be directly comparable, and we can note a significant spread between the two curves in the short and most liquid end of the maturity spectrum, which warrants the use of a multi-curve model to price nominal products. Finally, we observe that the option-implied YoY swap rates are close to the ZC swap rates. This relation implies only small levels of the convexity correction as seen directly in the right-hand-side of Figure~\ref{RIf:Discrete rates} where the convexity correction, as described in Section~\ref{RISecFundProd}, is plotted for different swap lengths.  
An implied lognormal volatility surface is constructed from the prices of these options (selecting out-of-the-money options where available) using a geometric Brownian motion model for the CPI index as described in Section~\ref{RISecComOthMod}. Two of the prices for the two year maturity are identically zero and are thus removed from the dataset. We find the at-the-money implied volatility of the YoY cap using the piecewise constant hermite interpolation. The surface is plotted in the left hand side of Figure~\ref{RIf:Vols}, and one can see a significant volatility smile, but also volatility levels that are quite low, around only 1.5-3\%. Finally we consider a EURIBOR term structure of swaptions with maturities ranging from 3m, 6m, 1Y, 2Y, 3Y, 5Y, 7Y, 10Y, 15Y, 20Y to 30Y. Since the focus of the paper is on the inflation component we limit our modelling to one curve---the 3m tenor curve. We thus calibrate only to swaptions with a one-year underlying swap length, since this swaption, by EUR market convention, contains payments involving only 3m EURIBOR. We refer to \citep{CrepeyEtAl2016} for a more extensive calibration involving matching the volatility of both the 3m and 6m EURIBOR curves in a rational model resembling this one, but without the inflation component.  Due to the lognormal assumption for swap rates precluding negative interest, rates it is now customary to quote swaption prices in normal or Bachelier implied volatility as opposed to lognormal. This is done on the right panel in Figure~\ref{RIf:Vols}. The data is from 1 January 2015, and we calibrate directly to these volatilities. 
\begin{figure}[H]
\centering
\includegraphics[scale=0.55]{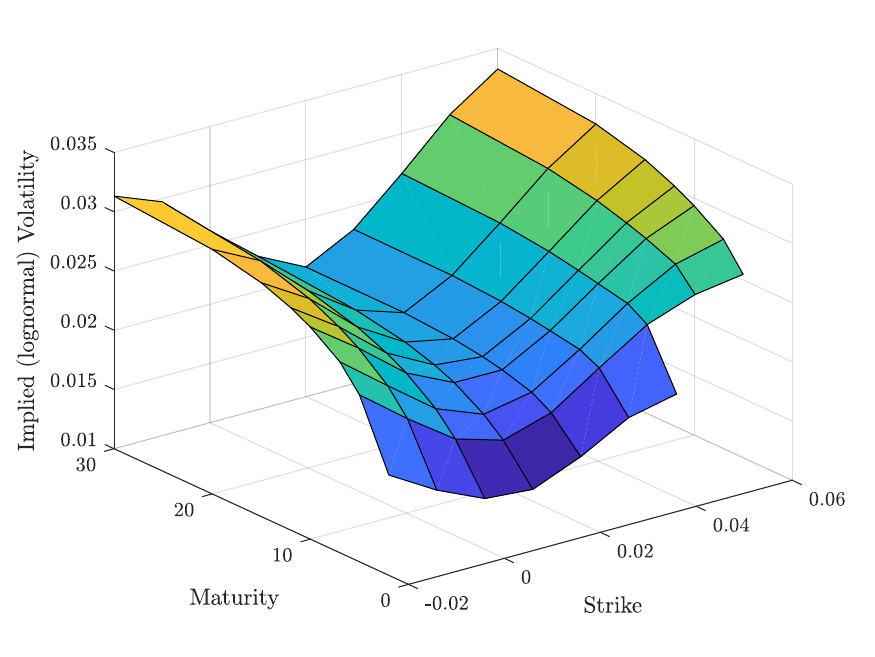}\includegraphics[scale=0.49]{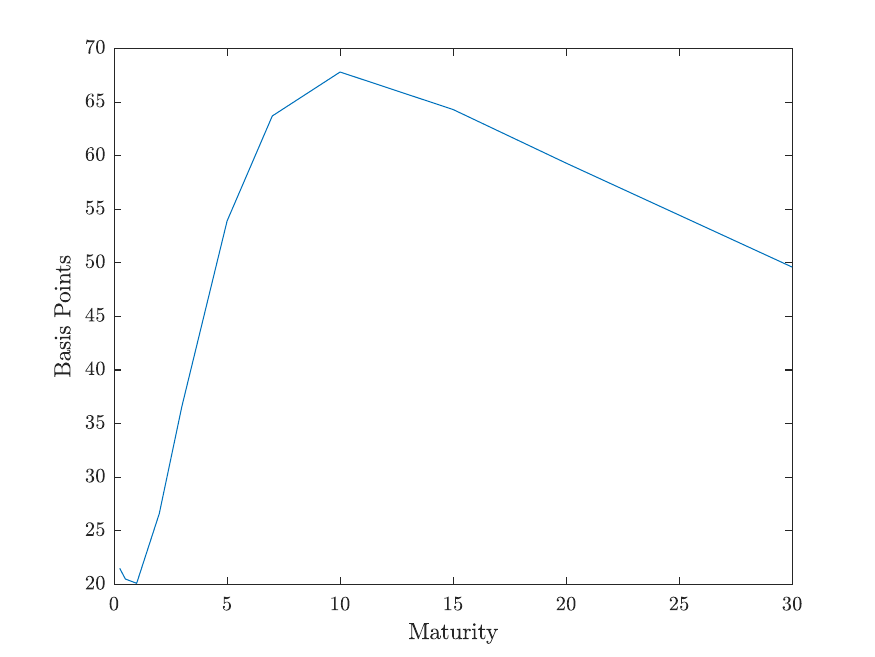}
\caption{On the left: Lognormal implied volatility surface. On the right: Implied normal (Bachelier) volatility in basis points for swaptions on 1Y swaps. The data is from 1 January 2015.}
\label{RIf:Vols}
\end{figure}
Very similar to Specification~\ref{RIExTimeChangeLevy}, our model setup is the following:
\begin{equation*}
X_t= \left  (X_t^{\suptext{R}} +\mu^{\suptext{R}}(t) ,X_{\tau^{\suptext{S}}(t)}^{\suptext{S}} + \mu^{\suptext{S}}(t), \mu^{\suptext{L}}(t) \right ),
\end{equation*}
where  $\mu^i(t)$ for $i=L,R,S$ are deterministic martingalizing functions and thus the model is a two factor model. We assume that $(X^{\suptext{R}}_t,X^{\suptext{S}}_t)$ is a two-dimensional {\Levy} process with independent marginals  and that $\tau^{\suptext{S}}(t)$ is a deterministic time-change. The two independent {\Levy} processes are defined by their Laplace exponents 
\begin{equation*}
\kappa_i(z)=\ln \left (\EM\left[\e^{z X^i_1}\right] \right ), \quad (i=\suptext{R,S}) 
\end{equation*}
at $t=1$. Thus we are in the additive exponential-rational pricing kernel setup with
\begin{equation*}
A_t^{\suptext{R}}=\e^{\innerp {w_{\subtext{R}} }{X_t}} ,\quad A_t^{\suptext{S}}=\e^{\innerp {w_{\subtext{S}}}{X_t}} ,\quad A_t^{\suptext{L}}=\e^{\innerp{w_{\subtext{L}}}{X_t}}.
\end{equation*}
We set  $w_{\subtext{S}}=(0,1,0)$, $w_{\subtext{R}} =(a_{\subtext{R}}   ,b a_{\subtext{R}} ,0)$ and $w_{\subtext{L}}=a_{\subtext{L}} w_{\subtext{R}}  + (0,0,1)$. This means that the $b$ parameter determines the dependence between the $(A^{\suptext{R}}_t)$ and $(A^{\suptext{S}}_t)$ and it furthermore means that the randomness in $(A^{\suptext{L}}_t)$ is merely a (log)-linear transformation of the randomness in $(A^{\suptext{R}}_t)$. As in \eqref{RIeqFundamentalMartingale}, when $w_\subtext{R},w_\subtext{S},w_\subtext{L} \in \assumExp(X)$ we can solve for the martingalizing drifts to obtain
\begin{align*}
\mu_{\subtext{S}}(t) =& -\tau^{\suptext{S}}(t)\kappa_{\subtext{S}}(1),\quad \mu_{\subtext{R}} (t) =  -\tau^{\suptext{S}}(t)\left( \frac{\kappa_{\subtext{S}}(a_{\subtext{R}} b)}{a_{\subtext{R}} }-b\kappa_{\subtext{S}}(1)\right)-t\kappa_{\subtext{R}}  (a_{\subtext{R}} ),\\
\mu_{\subtext{L}}(t)=& -\tau^{\suptext{S}}(t) \left (\kappa_{\subtext{S}}(a_{\subtext{L}} a_{\subtext{R}} b)-a_{\subtext{L}}\kappa_{\subtext{S}}(a_{\subtext{R}} b) \right )  -t \left ( \kappa_{\subtext{R}} (a_{\subtext{L}} a_{\subtext{R}} )-a_{\subtext{L}}\kappa_{\subtext{R}} (a_{\subtext{R}} ) \right ).
\end{align*}
We set the deterministic time-change $\tau^{\suptext{S}}(t) =\int_0^t a(s) \diffIN s$, where $a(t)$ is a piecewise constant function
\begin{equation*}
a(t) = a_k,  t\in (T_{k-1},T_k].
\end{equation*}
Here $\{T_0,T_1,\dots,T_8\}=\{0,2,5,7,10,12,15,20,30\}$, is the set of maturities quoted in the YoY option market. We calibrate the constants $a_1,\dots,a_8$ starting from the smallest to the largest maturity by matching to the YoY cap/floor volatility surface allowing a perfect fit to at least one strike per maturity. The dependence structure between the $\mathrm{R}$ and $\mathrm{S}$ component is fully determined by the parameter $b$ thus reducing the model to a two-factor setup where the calculated expressions for YoY caplets, YoY swap prices and swaption prices can be applied directly without approximation. 

The nominal and the real curve are fitted by construction, but fitting the term-structure of YoY swap rates is less straightforward, since the swap rate depends on the full parameter set of the model, see Eq.~\eqref{RIEqYoYFairRate}. We choose to calibrate the $b^{\suptext{R}}(t)$ function to this term structure. There is enough flexibility in the $b^{\suptext{R}}(t)$ function to fit the YoY swap rates without error, but direct calibration results in a quite volatile $b^{\suptext{R}}(t)$ function which is hardly desirable. Therefore we instead fit an eight-knot Hermite polynomial with a non-smoothness penalty -- a similar choice is made in \citep{GretarssonEtAl2012} -- and we find that the loss of accuracy when doing this is insignificant. The flexible shape means that the correlation parameter $b$ and volatility parameter $a_{\subtext{R}} $ in practice cannot be identified simultaneously with $b^{\suptext{R}}(t)$ from the YoY swap curve. We solve this issue by simply fixing the $b$ and the $a_{\subtext{R}} $ parameters before calibration. In practice one needs only to avoid setting these parameters too low because the convexity correction becomes zero, by construction, if $b=0$ or $a_{\subtext{R}} =0$. In both of our calibration examples we fix these values at $b=30$ and $a_{\subtext{R}} =0.25$. 

Since swap rates are determined not only by the $b^{\suptext{R}}(t)$ function, but the full parameter set of the model, one cannot calibrate $b^{\suptext{R}} (t)$ independently of $a(t)$ and the parameters determining the $(X_t)$ process. On the other hand, YoY cap and floor prices are primarily affected by the $(A_t^{\suptext{S}})$ component and thus not very sensitive to the changes in values of $b^{\suptext{R}}(t)$ unless the correlation between $(X_t^{\suptext{R}})$ and $(X_t^{\suptext{S}})$ is very high, which means this dual identification problem is in fact easily solved in practice. The overall calibration algorithm can be reduced to:
\begin{enumerate}
\item Set $b^{\suptext{R}}(t)=1$, and calibrate $a_1,\dots,a_8$ and the parameters of determining the law of $(X_t^{\suptext{S}})$ to YoY cap/floor implied volatilities.
\item Calibrate $b^{\suptext{R}}(t)$ to the curve of YoY swap rates rates using least squares minimization with a penalty for $b_{\subtext{R}} (t)\notin(0,1)$.
\item Repeat Step 1 using instead the updated values of $b^{\suptext{R}}(t)$. 
\item Calibrate  $b_{\subtext{L}}(\cdot,\cdot)$ to swaption prices.
\end{enumerate}
The swaption calibration is done by calibrating the $b_{\subtext{L}}(\cdot,\cdot)$ function sequentially. This means that the $a_{\subtext{L}}$ parameter, which also determines overall variance cannot be calibrated at the same time and we therefore fix it at $a_{\subtext{L}}=1.3$ below. We parametrise the $b_{\subtext{L}}(\cdot,\cdot)$ function by setting $b_{\subtext{L}}(t,t+3m)=P^{\suptext{N}}(0,t)L(0,t,t+3m)+\tilde{b}_{\subtext{L}}(t)$. The function $\tilde{b}_{\subtext{L}}(t)$ is then piecewise constant in relation to the swaption maturities we can observe, i.e.
\begin{equation*}
\tilde{b}_{\subtext{L}}(t)=\tilde{b}_k,  t\in (T_{k},T_{k+1}]
\end{equation*} 
with $\{T_0,T_1,\dots,T_{11}\}=\{3m,6m,1Y,2Y,3Y,5Y,7Y,10Y,15Y,20Y,30Y,31Y\}$.
\subsection*{Gaussian example}
We first assume that $(X^{\suptext{R}}_t)$ and $(X^{\suptext{S}}_t)$ are independent standard Brownian motions with Laplace exponent at $t=1$ $\kappa_{\subtext{R}}(z)=\kappa_{\subtext{S}} (z)=\frac{1}{2} z^2$, i.e. in the spirit of Specification~\ref{RIExTimeChangeWP}. We would not expect a Gaussian or log-normal model to be well suited to reproduce implied volatility smiles, but we nevertheless believe that a Gaussian setup is illustrative as a benchmark case of study.

As discussed above, we first fix the $b=30$ and $a_{\subtext{R}} =0.25$, and then proceed with the calibration algorithm described above. In Step 1 we choose to calibrate $a_1,\dots,a_8$ to at-the-money implied volatility. This is done sequentially starting with calibrating $a_1$ to the two-year YoY implied volatility and $a_2$ to the five-year YoY implied volatility, and so forth. We note that the value of $a_1$ affects not just the two-year maturity but all YoY option maturities (larger than two years) since we are calibrating directly to caps, which have annual payments every year until maturity. Thus the sequential nature of the calibration of these parameters is key. 
The result of this calibration can be seen in Figure~\ref{RIf:GaussianATM} where we plot at-the-money implied volatility from the market. We furthermore plot an example of the model smile in the for a fixed maturity of five years. While the model smile is not completely flat, the Gaussian structure is, by construction, not suited for smile fitting. Finally the YoY swap rates are fitted without error by adjusting the $b_{\subtext{R}} (t)$ function. YoY swap rates requires only mild adjustment of the $b_{\subtext{R}} (t)$ function away from its default value of one.\footnote{ All parameter values are available upon request. }
\begin{figure}[!htbp]
\centering
\includegraphics[scale=0.4]{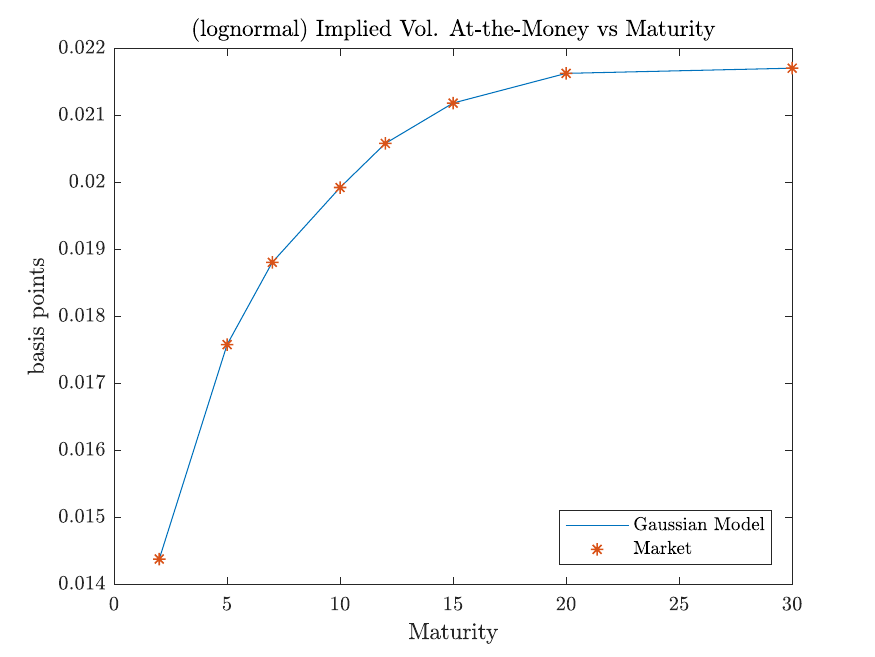}
\includegraphics[scale=0.4]{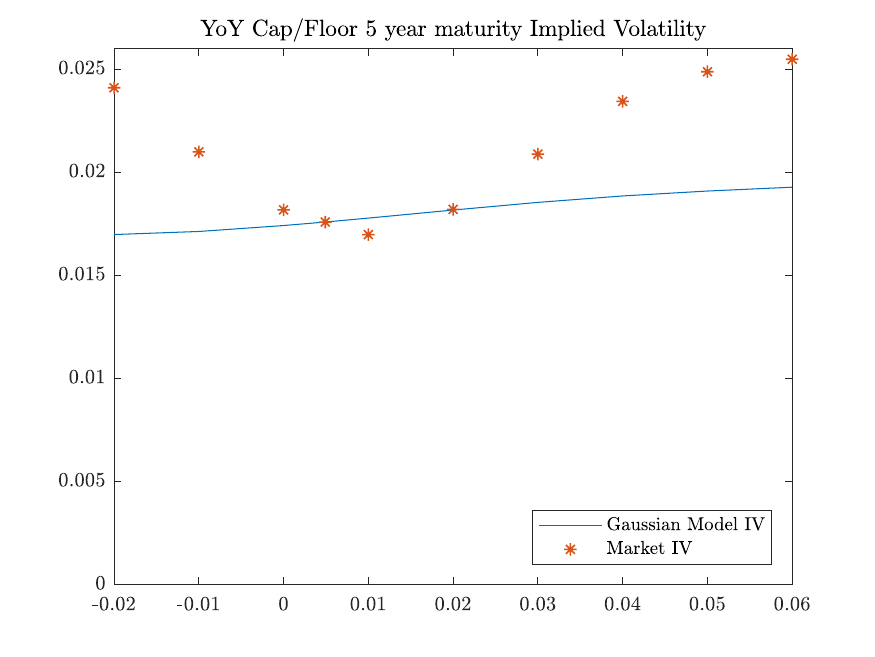}\\
\includegraphics[scale=0.4]{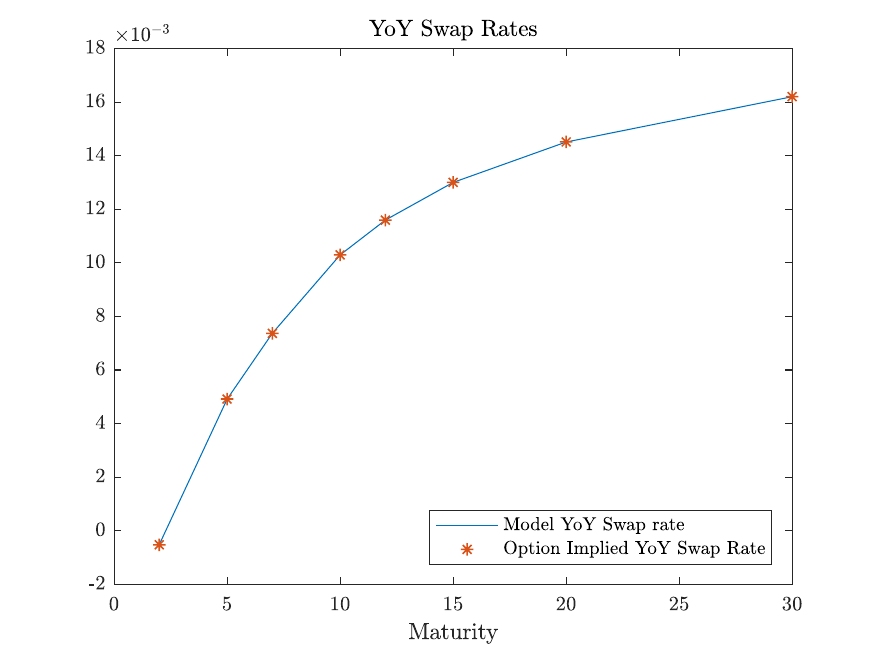}
\includegraphics[scale=0.4]{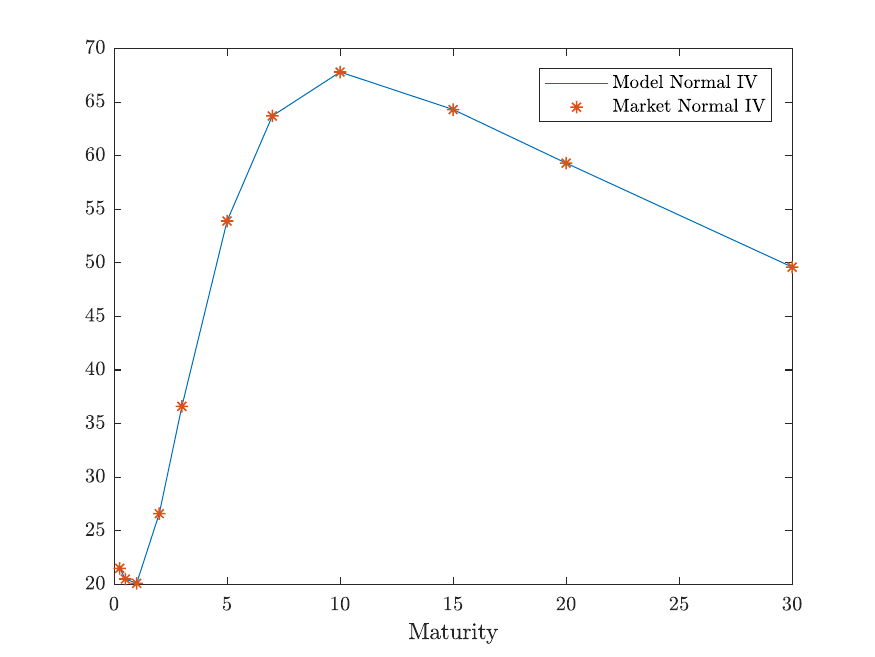}
\caption{\footnotesize Upper Left: At-the-money lognormal implied volatility of YoY caps model vs market using data from 1 January 2015. Upper Righ: YoY cap/floor implied volatility. Lower Left: YoY option implied swap rates, market vs. (Gaussian) model. Lower Right: At-the-money normal swaption implied volatility. 1st of January 2015.}
\label{RIf:GaussianATM}
\end{figure}
 
When fitting to swaptions we fix $a_{\subtext{L}}=1.3$ as explained above. Then we sequentially fit the $b_{\subtext{L}}(t)$ function directly to swaption normal implied volatility starting from the three-month maturity up to the thirty-year maturity. The model is made to fit at-the-money, swaptions only, and the results are plotted in the lower right quadrant of Figure~\ref{RIf:GaussianATM}. 

\subsection*{NIG example}
 To produce a model more in line with the volatility smile, we instead assume that $(X^{\suptext{R}}_t,X^{\suptext{S}}_t)$ are independent Normal Inverse Gaussian (NIG) processes, see for example \citep{Barndorff-Nielsen98}. We have that the Laplace exponent at $t=1$ is given by 
\begin{equation*}
\kappa_i (z) = -\nu_i \left(\sqrt{\nu_i^2 - 2 z\theta_i - z^2 \sigma_i^2} -\nu_i\right), i=R,S,
\end{equation*}
expressed in terms of the parametrisation $(\nu_i  , \theta_i ,\sigma_i )$ as
where $\nu_i , \sigma_i  > 0$ and $\theta_i \in \R.$ Since we want to control variance primarily using the time-change $\tau^{\suptext{S}}(t)$, we set $\sigma_i = \sqrt{1-\theta_i^2/\nu_i^2}$ so that $X^{\suptext{R}}_1$ and $X^{\suptext{S}}_1$ both have variance of 1. Since we are only calibrating to the YoY cap/floor smile the full distribution of both marginals in  $(X_t^{\suptext{R}}, X_t^{\suptext{S}})$ is not identified by the data. For simplicity, we also set $\nu_{\subtext{S}}=\nu_{\subtext{R}} $ and $\theta_{\subtext{S}}=\theta_{\subtext{R}} $. As in the Gaussian case we prefix  $b=30$ and  $a_{\subtext{R}} =0.25$. In the NIG case we split Step 1 in the calibration process by first fixing the rate of time at a constant, i.e. $a_i=a$, and then calibrate $a,\nu_{\subtext{S}},\theta_{\subtext{S}}$ to the whole YoY cap/floor implied volatility surface using the lsqnonlin algorithm in Matlab. Thereafter, the individual $a_1,\dots,a_8$ are calibrated sequentially such that the model fits the at-the-money implied volatilities without error. The rest of the algorithm is followed exactly like in the Gaussian case.
\begin{figure}[H]
\centerfloat
\includegraphics[scale=0.75]{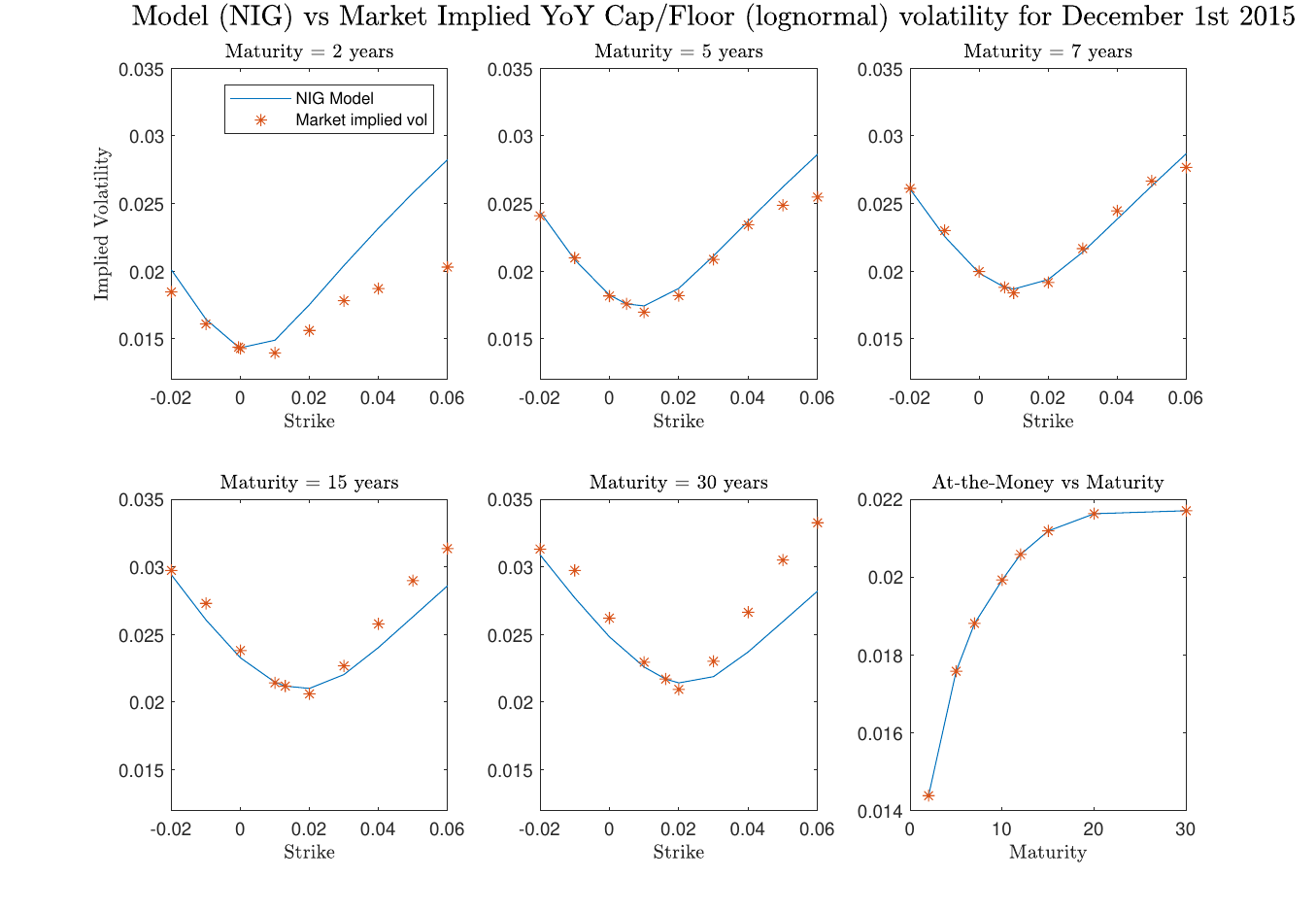}
\caption{Lognormal implied volatility for YoY cap/floors model vs market using data from 1 January 2015.}
\label{RIf:NIGYoYIV}
\end{figure}
The fit to YoY swap rates is indistinguishable from the graph in Figure~\ref{RIf:GaussianATM} and is not plotted again. We set $a_{\subtext{L}}=1.3$ and fit the $b_{\subtext{L}}(t)$ function to the same dataset of swaptions on one-year underlying swaps. The resulting fit is again indistinguishable from the fit in Figure~\ref{RIf:GaussianATM} and the calibrated $b_{\subtext{L}}(t)$ function is available upon request.
In Figure~\ref{RIf:NIGYoYIV}, we plot model vs market volatility smiles for select maturities. We have only used one time-dependent scaling, so the model fits the at-the-money level without error. The remaining option prices are in principle fitted using only two parameters $\nu_{\subtext{S}}$ and $\theta_{\subtext{S}}$. Thus we would not expect a perfect fit for all maturities. In general, any {\Levy} process is well known to exhibit a flattening smile as maturities are increased which often results in a slightly too steep smile in the short end and too flat in the long end. These problems could be resolved by introducing further time-inhomogeneity or by applying stochastic time-changes, but with the virtue of model simplicity taken into account, we view the calibrated setup as satisfactory.
\section{Conclusions}
This paper focuses primarily on the theoretical development of stochastic, rational term-structure models using pricing kernels suitable for the pricing of nominal and inflation-linked financial instruments. We demonstrate how this model class can be constructed with a view towards calibration to market data. We furthermore show how the models extend the classical short rate approach to inflation modelling. We expect future research to be focused more on the numerics of risk management within the model as well as calibration to a broader set of market instruments such as joint calibration of year-on-year and zero-coupon caps, as well as including time-series information in the calibration problem. 
\section*{Appendix}
\appendix
\section{Lemmas}
This section contains a number of lemmas used for the derivation of the formulae for option pricing. 
\begin{Lemma} \label{RILemKey3} 
Let $q_t(z) = \E_t [ \exp ( \innerp{z}{Y})]$ be the moment generating function of $Y = [Y_1,Y_2]$, a random vector with conditional distribution $m_t$. Assume $R>0$, 
$q(-R,1)<\infty$ and $\alpha >0$. Then, 
\begin{equation*}
\E_t \left [ \e^{Y_2} (1- \alpha \e^{Y_1 }  )^+ \right ]	=  \frac{1}{ \pi} \int_{\R^+} \Re \frac{ \alpha ^{-(R+ i u) }   q_t(-(R+iu),1)}{(R+i u) (1+R+iu)}   \diff u. 
\end{equation*} 
\end{Lemma}
We omit the proof of this lemma since it is standard.
\begin{Lemma} \label{RILemKey2} 
Let $q_t(z) = \E_t [ \exp ( \innerp{z}{Y})]$ be the moment generating function of $Y = [Y_1,Y_2]$, a random vector with conditional distribution $m_t$. Assume $R<-1$, $q_t(-R,1)<\infty$ and $\alpha >0 $. Then, 
\begin{equation*}
\E_t \left [ \e^{Y_2} (\alpha \e^{Y_1}  -1   )^+ \right ]	=  \frac{1}{ \pi} \int_{\R^+} \Re \frac{  \alpha ^{-(R+ i u) }   q_t(-(R+iu),1)}{(R+i u) (1+R+iu)} \diff u. 
\end{equation*} 
\end{Lemma}
\begin{Lemma} \label{RILemKey}
Let $q_t(z) = \E_t [ \exp ( \innerp{z}{Y})]$ be the conditional moment generating function of $Y = [Y_1,Y_2,Y_3]$, a random vector with conditional distribution $m_t$. Assume $R>0$ and that 
$
q_t(-R,1,1)+q_t(-R,1,0)<\infty.
$
Then 
\begin{equation*}
\E_t \left [ (1+ \e^{Y_3} ) \e^{Y_2} ( 1- \e^{Y_1 }  )^+ \right ]	=  \frac{1}{ \pi} \int_{\R^+} \Re \frac{  q_t(-(R+iu),1,0)+ q_t(-(R+iu),1,1) }{(R+i u) (1+R+iu)}   \diff u 
\end{equation*} 
\end{Lemma}
\begin{proof}
Write 
\begin{equation*}
\E_t \left [ (1+ \e^{Y_3} ) \e^{Y_2} ( 1- \e^{Y_1 }  )^+ \right ] = \E_t \left [ \e^{Y_2} ( 1- \e^{Y_1 }  )^+ \right ]  + \E_t \left [  \e^{Y_3 + Y_2} ( 1- \e^{Y_1 }  )^+ \right ] 
\end{equation*}
 and apply Lemma~\ref{RILemKey3} to each term. 
\end{proof}
\begin{Lemma} \label{RILemBS}
Let $X \sim N(\mu_x, \sigma_x^2)$ and $Y\sim N(\mu_y,\sigma_y^2)$  with $\Cov[X,Y] = \sigma_{xy}$ and assume that $\alpha > 0$. Then
\begin{equation*}
\E\left[\e^X ( \alpha \e^Y -1)^+ \right]  = \alpha  \e^{ \delta + \mu_y + \tfrac{1}{2} ( 1+ 2b ) \sigma_y^2 } N(  d +\sigma_y ) -   \e^{ \delta  } N (    d  )  
 \end{equation*}
 where $a = \mu_x - \frac{\sigma_{xy}}{\sigma_y^2}  \mu_y+ (\sigma_x^2-\sigma^2_{xy}/\sigma_y^2)/2$, $b=\sigma_{xy}/\sigma_y^2$, $\delta= a + b \mu_y + (b \sigma_y)^2/2$ and $d = \tfrac{1}{\sigma_y}  (  \ln \alpha+ b\sigma_y^2 + \mu_y) $. 
\end{Lemma}
\begin{proof}
By conditional of normals  $\E[\e^X \given   Y] = \e^{a+bY}$. Noting that $-\tfrac{ (x-\mu)^2}{2 \sigma^2 } + bx   = b \mu +\tfrac{b^2 \sigma^2}{2} -\tfrac{ (x-\mu- b \sigma^2)^2}{2 \sigma^2 }  $, the tower property yields
\begin{align*}
\E \left [\e^X ( \alpha \e^Y  - 1 )^+ \right  ]  = \E \left [\e^{a+bY} ( \alpha \e^Y -1) 1 _{ \{Y \geq \ln  \tfrac{1}{\alpha} \} } \right ]  =  \alpha  \e^{ \delta + \mu_y +  \tfrac{1}{2} ( 1+ 2b ) \sigma_y^2 } N(  d +\sigma_y ) -   \e^{ \delta  } N (    d  ),  
\end{align*}
as sought.
\end{proof}
\begin{Lemma} \label{RILemBS2}
Let $X,Y$ and $\alpha$ be as in  Lemma~\ref{RILemBS}. Then
\begin{equation*}
\E\left[\e^X ( 1 - \alpha \e^Y )^+ \right]  =    \e^{ \delta  } N (-d)   - \alpha  \e^{ \delta + \mu_y + \tfrac{1}{2} ( 1+ 2b ) \sigma_y^2 } N\left(-d-\sigma_y \right) 
 \end{equation*}
where $a,b, \delta$ and $d$ are as in Lemma~\ref{RILemBS}. 
\end{Lemma}
\begin{proof}
As Lemma~\ref{RILemBS}.
\end{proof}

\bibliography{bibliography/bibliography}

\begin{thebibliography}{}

\bibitem[Barndorff-Nielsen, 1998]{Barndorff-Nielsen98}
Barndorff-Nielsen, O.~E. (1998).
\newblock Processes of normal inverse {Gaussian} type.
\newblock {\em Finance and Stochastics}, 2(1):41--68.
\newblock \url{https://doi.org/10.1007/s007800050032}.

\bibitem[Belgrade et~al., 2004]{BelgradeEtAl2004}
Belgrade, N., Benhamou, E., and Koehler, E. (2004).
\newblock A market model for inflation.
\newblock {\em Working Paper}.
\newblock \url{http://dx.doi.org/10.2139/ssrn.576081}.

\bibitem[{Bj{\"o}rk}, 2009]{Bjork2009}
{Bj{\"o}rk}, T. (2009).
\newblock {\em Arbitrage Theory in Continuous Time}.
\newblock Oxford University Press.

\bibitem[Black and Litterman, 1992]{BlackLitterman1992}
Black, F. and Litterman, R. (1992).
\newblock Global portfolio optimization.
\newblock {\em Financial Analysts Journal}, 48(5):28--43.
\newblock \url{https://doi.org/10.2469/faj.v48.n5.28}.

\bibitem[Black and Scholes, 1973]{BlackScholes1973}
Black, F. and Scholes, M. (1973).
\newblock The pricing of options and corporate liabilities.
\newblock {\em Journal of Political Economy}, 81(3):637--654.
\newblock \url{http://dx.doi.org/10.1086/260062}.

\bibitem[Brigo and Mercurio, 2007]{BrigoMercurio2007}
Brigo, D. and Mercurio, F. (2007).
\newblock {\em Interest Rate Models -- Theory and Practice: With Smile,
  Inflation and Credit}.
\newblock Springer-Verlag Berlin Heidelberg.

\bibitem[Constantinides, 1992]{Constantinides1992}
Constantinides, G.~M. (1992).
\newblock A theory of the nominal term structure of interest rates.
\newblock {\em The Review of Financial Studies}, 5(4):531--552.
\newblock \url{https://doi.org/10.1093/rfs/5.4.531}.

\bibitem[Cr{\'e}pey et~al., 2016]{CrepeyEtAl2016}
Cr{\'e}pey, S., Macrina, A., Nguyen, T.~M., and Skovmand, D. (2016).
\newblock Rational multi-curve models with counterparty-risk valuation
  adjustments.
\newblock {\em Quantitative Finance}, 16(6):847--866.
\newblock \url{https://doi.org/10.1080/14697688.2015.1095348}.

\bibitem[Cuchiero et~al., 2019]{CuchieroEtAl2019}
Cuchiero, C., Fontana, C., and Gnoatto, A. (2019).
\newblock Affine multiple yield curve models.
\newblock {\em Mathematical Finance}, 29(2):568--611.

\bibitem[Cuchiero et~al., 2016]{CKT2016}
Cuchiero, C., Klein, I., and Teichmann, J. (2016).
\newblock A new perspective on the fundamental theorem of asset pricing for large financial markets.
\newblock {\em Theory of Probability \& Its Applications }, 60(4):561--579.

\bibitem[Dodgson and Kainth, 2006]{KainthDodgson2006}
Dodgson, M. and Kainth, D. (2006).
\newblock Inflation-linked derivatives.
\newblock {\em Royal Bank of Scotland Risk Training Course, Market Risk Group}.

\bibitem[Döberlein and Schweizer, 2001]{DoberleinSchweizer2001}
Döberlein, F. and Schweizer, M. (2001).
\newblock On savings accounts in semimartingale term structure models.
\newblock {\em Stochastic Analysis and Applications}, 19(4):605--626.

\bibitem[Filipovi{\'c} et~al., 2017]{FilipovicLarssonTrolle2017}
Filipovi{\'c}, D., Larsson, M., and Trolle, A.~B. (2017).
\newblock Linear-rational term structure models.
\newblock {\em The Journal of Finance}, 72(2):655--704.
\newblock \url{https://doi.org/10.1111/jofi.12488}.

\bibitem[Flesaker and Hughston, 1996a]{FlesakerHughston1996a}
Flesaker, B. and Hughston, L. (1996a).
\newblock Positive interest.
\newblock {\em Risk}, 9(1):46--49.

\bibitem[Flesaker and Hughston, 1996b]{FlesakerHughston1996b}
Flesaker, B. and Hughston, L. (1996b).
\newblock Positive interest: Foreign exchange.
\newblock In {\em Vasicek and Beyond}. Risk Publication.

\bibitem[Grbac et~al., 2015]{GrbacEtAl2015}
Grbac, Z., Papapantoleon, A., Schoenmakers, J., and Skovmand, D. (2015).
\newblock Affine {LIBOR} models with multiple curves: Theory, examples and
  calibration.
\newblock {\em {SIAM} Journal on Financial Mathematics}, 6(1):984--1025.
\newblock \url{https://doi.org/10.1137/15M1011731}.

\bibitem[Grbac and Runggaldier, 2015]{GrbacRunggaldier2016}
Grbac, Z. and Runggaldier, W. (2015).
\newblock {\em Interest Rate Modeling: Post-Crisis Challenges and Approaches}.
\newblock Springer.

\bibitem[Gretarsson et~al., 2012]{GretarssonEtAl2012}
Gretarsson, H., Ribeiro, D., and Trovato, M. (2012).
\newblock Quadratic gaussian inflation.
\newblock {\em Risk}, 25(9):86.

\bibitem[Heston, 1993]{Heston1993}
Heston, S.~L. (1993).
\newblock A closed-form solution for options with stochastic volatility with
  applications to bond and currency options.
\newblock {\em The Review of Financial Studies}, 6(2):327--343.
\newblock \url{https://doi.org/10.1093/rfs/6.2.327}.

\bibitem[Hinnerich, 2008]{Hinnerich2008}
Hinnerich, M. (2008).
\newblock Inflation-indexed swaps and swaptions.
\newblock {\em Journal of Banking \& Finance}, 32(11):2293--2306.
\newblock \url{https://doi.org/10.1016/j.jbankfin.2007.04.033}.

\bibitem[Hoyle et~al., 2011]{HoyleEtAl2011}
Hoyle, E., Hughston, L.~P., and Macrina, A. (2011).
\newblock {L{\'e}vy} random bridges and the modelling of financial information.
\newblock {\em Stochastic Processes and their Applications}, 121(4):856--884.
\newblock \url{https://doi.org/10.1016/j.spa.2010.12.003}.

\bibitem[Hughston, 1998]{Hughston1998}
Hughston, L. (1998).
\newblock Inflation derivatives.
\newblock {\em Merrill Lynch and King's College London Working Paper}.

\bibitem[Hull and White, 1990]{HullWhite1990}
Hull, J. and White, A. (1990).
\newblock Pricing interest-rate-derivative securities.
\newblock {\em The Review of Financial Studies}, 3(4):573--592.
\newblock \url{https://doi.org/10.1093/rfs/3.4.573}.

\bibitem[Hunt and Kennedy, 2004]{HuntKennedy2004}
Hunt, P. and Kennedy, J. (2004).
\newblock {\em Financial Derivatives in Theory and Practice}.
\newblock Wiley.

\bibitem[Hurd and Zhou, 2010]{HurdZhou2010}
Hurd, T.~R. and Zhou, Z. (2010).
\newblock A {Fourier} transform method for spread option pricing.
\newblock {\em {SIAM} Journal on Financial Mathematics}, 1(1):142--157.
\newblock \url{https://doi.org/10.1137/090750421}.

\bibitem[Jacod and Shiryaev, 2003]{JacodShiryaev2003}
Jacod, J. and Shiryaev, A.~N. (2003).
\newblock {\em Limit Theorems for Stochastic Processes}.
\newblock Springer-Verlag Berlin Heidelberg.

\bibitem[Jarrow and Yildirim, 2003]{JarrowYildirim2003}
Jarrow, R. and Yildirim, Y. (2003).
\newblock Pricing treasury inflation protected securities and related
  derivatives using an {HJM} model.
\newblock {\em Journal of Financial and Quantitative Analysis}, 38(2):337--358.
\newblock \url{ https://doi.org/10.2307/4126754}.

\bibitem[Keller-Ressel et~al., 2013]{Keller-ResselEtAl2013}
Keller-Ressel, M., Papapantoleon, A., and Teichmann, J. (2013).
\newblock The affine {LIBOR} models.
\newblock {\em Mathematical Finance}, 23(4):627--658.
\newblock \url{https://doi.org/10.1111/j.1467-9965.2012.00531.x}.

\bibitem[Kenyon, 2008]{Kenyon2008}
Kenyon, C. (2008).
\newblock Inflation is normal.
\newblock {\em Risk}.

\bibitem[Kim, 2014]{Kim2014}
Kim, D.~H. (2014).
\newblock Swaption pricing in affine and other models.
\newblock {\em Mathematical Finance}, 24(4):790--820.
\newblock \url{https://doi.org/10.1111/mafi.12014}.

\bibitem[Korn and Kruse, 2004]{KornKruse2004}
Korn, R. and Kruse, S. (2004).
\newblock Einfache verfahren zur bewertung von inflationsgekoppelten
  finanzprodukten.
\newblock {\em Bl{\"a}tter der {DGVFM}}, 26(3):351--367.
\newblock \url{https://doi.org/10.1007/BF02808386}.

\bibitem[Kruse, 2011]{Kruse2011}
Kruse, S. (2011).
\newblock On the pricing of inflation-indexed caps.
\newblock {\em European Actuarial Journal}, 1(2):379--393.

\bibitem[Macrina, 2014]{Macrina2014}
Macrina, A. (2014).
\newblock Heat kernel models for asset pricing.
\newblock {\em International Journal of Theoretical and Applied Finance},
  17(07):1450048.
\newblock \url{https://doi.org/10.1142/S0219024914500484}.

\bibitem[Macrina and Mahomed, 2018]{MacrinaMahomed2018}
Macrina, A. and Mahomed, O. (2018).
\newblock Consistent valuation across curves using pricing kernels.
\newblock {\em Risks}, 6(1):18.

\bibitem[Mercurio, 2005]{Mercurio2005}
Mercurio, F. (2005).
\newblock Pricing inflation-indexed derivatives.
\newblock {\em Quantitative Finance}, 5(3):289--302.
\newblock \url{https://doi.org/10.1080/14697680500148851}.

\bibitem[Mercurio and Moreni, 2006]{MercurioMoreni2006}
Mercurio, F. and Moreni, N. (2006).
\newblock Inflation with a smile.
\newblock {\em Risk}, 19(3):70--75.

\bibitem[Mercurio and Moreni, 2009]{MercurioMoreni2009}
Mercurio, F. and Moreni, N. (2009).
\newblock A multi-factor {SABR} model for forward inflation rates.
\newblock {\em \url{ http://dx.doi.org/10.2139/ssrn.1337811}}.

\bibitem[Nguyen and Seifried, 2015]{NguyenSeifried2015}
Nguyen, T.~A. and Seifried, F.~T. (2015).
\newblock The multi-curve potential model.
\newblock {\em International Journal of Theoretical and Applied Finance},
  18(07):1550049.

\bibitem[Ribeiro, 2013]{Ribeiro2013}
Ribeiro, D. (2013).
\newblock A local volatility model to price and calibrate year-on-year and
  zero-coupon inflation options.
\newblock {\em \url{http://dx.doi.org/10.2139/ssrn.2290730}}.

\bibitem[Rogers, 1997]{Rogers1997}
Rogers, L. (1997).
\newblock The potential approach to the term structure of interest rates and
  foreign exchange rates.
\newblock {\em Mathematical Finance}, 7(2):157--176.
\newblock \url{https://doi.org/10.1111/1467-9965.00029}.

\bibitem[Rubinstein, 1991]{Rubinstein1991}
Rubinstein, M. (1991).
\newblock Pay now, choose later.
\newblock {\em Risk}, 4(2):13.

\bibitem[Ruijter and Oosterlee, 2012]{RuijterOosterlee2012}
Ruijter, M.~J. and Oosterlee, C.~W. (2012).
\newblock Two-dimensional {Fourier} cosine series expansion method for pricing
  financial options.
\newblock {\em {SIAM} Journal on Scientific Computing}, 34(5):B642--B671.
\newblock \url{ https://doi.org/10.1137/120862053}.

\bibitem[Rutkowski, 1997]{Rutkowski1997}
Rutkowski, M. (1997).
\newblock A note on the {Flesaker}-{Hughston} model of the term structure of
  interest rates.
\newblock {\em Applied Mathematical Finance}, 4(3):151--163.
\newblock \url{https://doi.org/10.1080/135048697334782}.

\bibitem[Sato, 1999]{Sato1999}
Sato, K.-i. (1999).
\newblock {\em {L{\'e}vy} Processes and Infinitely Divisible Distributions}.
\newblock Cambridge University Press.

\bibitem[Singleton and Umantsev, 2002]{SingletonUmantsev2002}
Singleton, K.~J. and Umantsev, L. (2002).
\newblock Pricing coupon-bond options and swaptions in affine term structure
  models.
\newblock {\em Mathematical Finance}, 12(4):427--446.
\newblock \url{https://doi.org/10.1111/j.1467-9965.2002.tb00132.x}.

\bibitem[Singor et~al., 2013]{SingorEtAl2013}
Singor, S.~N., Grzelak, L.~A., Van~Bragt, D.~D., and Oosterlee, C.~W. (2013).
\newblock Pricing inflation products with stochastic volatility and stochastic
  interest rates.
\newblock {\em Insurance: Mathematics and Economics}, 52(2):286--299.
\newblock \url{https://doi.org/10.1016/j.insmatheco.2013.01.003}.

\bibitem[Waldenberger, 2017]{Waldenberger2017}
Waldenberger, S. (2017).
\newblock The affine inflation market models.
\newblock {\em Applied Mathematical Finance}, pages 1--21.
\newblock \url{https://doi.org/10.1080/1350486X.2017.1378582}.

\end{thebibliography}

\end{document}